\documentclass[a4paper,USenglish,numberwithinsect]{article}

\usepackage[utf8]{inputenc}
\usepackage{fullpage}
\usepackage{microtype}
\usepackage{amsfonts}
\usepackage{amssymb}
\usepackage{amsmath}
\usepackage{amsthm}
\usepackage{paralist}
\usepackage{graphicx}
\usepackage{verbatim}
\usepackage{xspace}
\usepackage{chngcntr}
\usepackage{authblk}
\usepackage{enumitem}
\usepackage[unicode]{hyperref}

\newtheorem{theorem}{Theorem}[section]
\newtheorem{lemma}[theorem]{Lemma}
\newtheorem{definition}[]{Definition}
\theoremstyle{remark}
\newtheorem*{remark}{Remark}

\newcommand{\mycase}[1]{\smallskip\noindent{\bf Case (#1):}}

\newcommand\eps\varepsilon

\newcommand{\etal}{{\em et al.\/}}
\newcommand{\N}{\mathbb{N}}
\newcommand{\Oh}{\mathcal{O}}
\newcommand{\caseref}[1]{{\rm \ref{#1}}}

\newcommand\rel{{\mbox{rel}}}
\newcommand\ALG{\textsf{ALG\/}\xspace}
\newcommand\DIV{\textsf{PG-DIV\/}\xspace}
\newcommand\MAIN{\textsf{PG\/}\xspace}
\newcommand\OPT{\textsf{OPT\/}\xspace}
\newcommand\ADV{\textsf{ADV\/}\xspace}
\newcommand\ALGs{{\textsf{ALG\/}(s)}\xspace}
\newcommand\DIVs{{\textsf{PG-DIV\/}(s)}\xspace}
\newcommand\MAINs{{\textsf{PG\/}(s)}\xspace}
\newcommand\unarylt{{<\,}}
\newcommand\unarygeq{{\geq\,}}

\numberwithin{equation}{section} 

\graphicspath{{./}} 

\title{On Packet Scheduling with Adversarial Jamming and Speedup\thanks{Work partially
supported by GA \v{C}R project 17-09142S, GAUK project 634217 and Polish National Science Center grant 2016/21/D/ST6/02402.
A preliminary version of this work is in~\cite{jamming-waoa}.}
}

\author[1]{Martin B\"{o}hm}
\author[2]{{\L}ukasz Je\.{z}}
\author[1]{Ji{\v{r}}\'{\i} Sgall}
\author[1]{Pavel Vesel\'{y}}
	
\affil[1]{Computer Science Institute of Charles University, Prague, Czech Republic\\
	 \texttt{\{bohm,sgall,vesely\}@iuuk.mff.cuni.cz}}
\affil[2]{Institute of Computer Science, University of Wroc{\l}aw, Poland\\
  \texttt{lje@cs.uni.wroc.pl}}

\date{}

\begin{document}

\maketitle

\begin{abstract}
In Packet Scheduling with Adversarial Jamming packets of arbitrary sizes
arrive over time to be transmitted over a channel
in which instantaneous jamming errors occur at times chosen by the adversary
and not known to the algorithm. The transmission taking place at the time
of jamming is corrupt, and the algorithm learns this fact immediately.
An online algorithm maximizes the total size of packets
it successfully transmits and the goal is to develop
an algorithm with the lowest possible asymptotic competitive
ratio, where the additive constant may depend on packet sizes.

Our main contribution is a universal algorithm that works for any speedup
and packet sizes and, unlike previous algorithms for the problem, it
does not need to know these parameters in advance. We show that this
algorithm guarantees 1-competitiveness with speedup 4, making it the
first known algorithm to maintain 1-competitiveness with a moderate
speedup in the general setting of arbitrary packet sizes. 
We also prove a lower bound of $\phi+1\approx 2.618$ on the speedup of any
1-competitive deterministic algorithm, showing that our algorithm is
close to the optimum.

Additionally, we formulate a general framework for analyzing our
algorithm locally and use it to show upper bounds on its competitive
ratio for speedups in $[1,4)$ and for several special cases,
recovering some previously known results, each of which had a dedicated proof.
In particular, our algorithm is 3-competitive without speedup, matching
both the (worst-case) performance of the algorithm by Jurdzinski~\etal~\cite{JurdzinskiKL13}
and the lower bound by Anta~\etal~\cite{Anta13}.
%
\end{abstract}

\section{Introduction}

We study an online packet scheduling model recently introduced by Anta
\etal~\cite{Anta13} and extended by Jurdzinski \etal~\cite{JurdzinskiKL13}.
In our model, packets of arbitrary sizes arrive over time and they are to be transmitted
over a single communication channel. The algorithm
can schedule any packet of its choice at any time, but cannot interrupt its
subsequent transmission; in the scheduling jargon, there is a single machine 
and no preemptions.  There are, however, instantaneous \emph{jamming errors} 
or \emph{faults} at times chosen by the adversary, which are not known to the
algorithm.  A transmission taking place at the time of jamming is
corrupt, and the algorithm learns this fact immediately. The packet
whose transmission failed can be retransmitted immediately or at any
later time, but the new transmission needs to send the whole packet, i.e.,
the algorithm cannot resume a transmission that failed.

The objective is to maximize the total size of packets successfully
transmitted.  In particular, the goal is to develop an online algorithm
with the lowest possible competitive ratio, which is the asymptotic worst-case ratio
between the total size of packets in an optimal offline schedule
and the total size of packets completed by the algorithm
on a large instance.
(See the next subsection for a detailed explanation of competitive analysis.)

We focus on algorithms with
resource augmentation, namely on online algorithms that transmit packets
$s\geq 1$ times faster than the offline optimum solution they are compared
against; such algorithm is often said to be speed-$s$, running at speed $s$,
or having a speedup of $s$.
As our problem allows constant competitive ratio already at speed $1$,
we consider the competitive ratio as a function of the speed.
This deviates from previous work, which focused on the case with
no speedup or on the speedup sufficient for ratio 1, ignoring 
intermediate cases.

\subsection{Competitive Analysis and its Extensions}

Competitive analysis focuses on determining the \emph{competitive ratio} of an
online algorithm.  The competitive ratio coincides with the approximation ratio,
i.e., the supremum over all valid instances $I$ of $\OPT(I)/\ALG(I)$, which is the
ratio of the optimum profit to the profit of an algorithm \ALG on
instance $I$.\footnote{We note that this ratio
is always at least $1$ for a maximization problem such as ours, but some authors
always consider the reciprocal, i.e., the ``alg-to-opt'' ratio, which is then at
most $1$ for maximization problems and at least $1$ for minimization problems.}
This name, as opposed to approximation ratio, is used for historical reasons,
and stresses that the nature of the hardness at hand is not due to computational
complexity, but rather the online mode of computation, i.e., processing an
unpredictable sequence of requests, completing each without knowing the future.
Note that the optimum solution is to the whole instance, so it can be thought of
as being determined by an algorithm that knows the whole instance in advance and
has unlimited computational power; for this reason, the optimum solution is
sometimes called ``offline optimum''.
Competitive analysis, not yet called this way, was first applied by Sleator and
Tarjan to analyze list update and paging problems~\cite{SleatorT85}.  Since then,
it was employed to the study of many online optimization problems, as evidenced
by (now somewhat dated) textbook by Borodin and El-Yaniv~\cite{BEY}.
A nice overview of competitive analysis and its many extensions in the scheduling
context, including intuitions and meaning can be found in a survey
by Pruhs~\cite{Pruhs07}.

\subsubsection{Asymptotic Ratio and Additive Constant}

In some
problems of discrete nature, such as bin packing or various coloring problems,
it appears that the standard notion of competitive analysis is too restrictive.
The problem is that in order to attain competitive ratio relatively close to $1$ (or even any ratio),
an online algorithm must behave in a predictable way when the
current optimum value is still small, which makes the algorithms
more or less trivial and the ratio somewhat large.  To remedy this, the
``asymptotic competitive ratio'' is often considered, which means essentially
that only instances with a sufficiently large optimum value are considered.
This is often captured by stating that an algorithm is $R$-competitive if
(in our convention) there exists a constant $c$ such that
$R \cdot \ALG(I) + c \geq \OPT(I)$ holds for every instance $I$.
The constant $c$ is typically required not to depend on the class of instances
considered, which makes sense for aforementioned problems where the optimum
value corresponds to the number of bins or colors used, but is still sometimes too
restrictive.

This is the case in our problem. Specifically, using an example we show that
a deterministic algorithm running at speed $1$ can be (constant) competitive
only if the additive term in the definition of the competitive ratio depends on the
values of the packet sizes, even if there are only two packet sizes.
Suppose that a packet of size $\ell$
arrives at time $0$. If the algorithm starts transmitting it immediately 
at time $0$, then at time $\varepsilon > 0$ a packet of size $\ell-2\epsilon$ arrives,
the next fault is at time $\ell-\varepsilon$ and then the schedule ends,
i.e., the time horizon is at $T = \ell-\varepsilon$.
Thus the algorithm does not complete the packet of size $\ell$, while 
the adversary completes a slightly smaller packet of size $\ell-2\epsilon$.
Otherwise, the algorithm is idle till some time $\varepsilon > 0$,
no other packet arrives and the next fault is at time $\ell$, which is also the time horizon.
In this case, the packet of size $\ell$ is completed in the optimal schedule,
while the algorithm completes no packet again.

\subsubsection{Resource Augmentation}

Last but not least, some problems do not admit competitive algorithms at all
or yield counterintuitive results.  Again, our problem is an example of the
former kind if no additive constant depending on packet sizes is allowed
(cf.~aforementioned example), whereas the latter can be observed in the paging
problem, where the optimum ratio equals the cache size, seemingly suggesting
that the larger the cache size, the worse the performance, regardless of the
caching policy.  Perhaps for this reason, already Sleator and Tarjan~\cite{SleatorT85} considered
\emph{resource augmentation} for paging problem, comparing an online algorithm
with cache capacity $k$ to the optimum with cache capacity $h \leq k$.
Yet again, they were ahead of the time: resource augmentation was re-introduced
and popularized in scheduling problems
by Kalyanasundaram and Pruhs~\cite{speed-clairvoyance}, and the name itself
coined by Phillips \etal~\cite{Phillips-time-critical}, who also considered
scheduling.

We give a brief overview of some of the work on resource augmentation
in online scheduling, focusing on interesting open problems.
As mentioned before, Kalyanasundaram and Pruhs~\cite{speed-clairvoyance}
introduced resource augmentation. Among other
results they proved that a constant competitive ratio is possible with
a constant speedup for a preemptive variant of real-time scheduling
where each job has a release time, deadline, processing time and a
weight and the objective is to maximize the weight of jobs completed
by their deadlines on a single machine.  Subsequently resource
augmentation was applied in various scenarios.  Of the most relevant
for us are those that considered algorithms with speedup that are
1-competitive, i.e., as good as the optimum. We mention two models
that still contain interesting open problems.

For real-time scheduling, Phillips \etal~\cite{Phillips-time-critical}
considered the underloaded case in
which there exists a schedule that completes all the jobs. It is easy
to see that on a single machine, the \emph{Earliest-Deadline First} (EDF)
algorithm is then an optimal online algorithm.  Phillips
\etal~\cite{Phillips-time-critical} proved that EDF on $m$ machines is
1-competitive with speedup $2-1/m$. (Here the weights are not
relevant.) Intriguingly, finding a 1-competitive algorithm with
minimal speedup for $m>1$ is wide open: It is known that speedup at least $6/5$
is necessary, it has been conjectured that speedup $e/(e-1)$ is
sufficient, but the best upper bound proven is $2-2/(m+1)$
from~\cite{Lam-tradeoffs}. See Schewior~\cite{Schewior-thesis} for
more on this problem.

Later these results were extended to real-time scheduling of
overloaded systems, where for uniform density (i.e., weight equal to
processing time) Lam \etal~\cite{Lam-EDF-overload} have shown that a
variant of EDF with admission control is 1-competitive with speedup 2
on a single machine and with speedup 3 on more machines. For
non-uniform densities, the necessary speedup is a constant if each job
is tight (its deadline equals its release time plus its processing
time)~\cite{Koo-tight-deadlines}. Without this restriction it is
no longer constant, depending on the ratio $\xi$ of the maximum and minimum
weight. It is known that it is at least $\Omega(\log\log \xi)$ and at
most $\Oh(\log \xi)$~\cite{Chrobak-overloaded,Lam-EDF-overload}; as far
as we are aware, closing this gap is still an open problem.

\subsection{Previous and Related Results}

The model was introduced by Anta \etal~\cite{Anta13}, who resolved it
for two packet sizes: If $\gamma>1$ denotes the ratio of the two sizes,
then the optimum competitive ratio for deterministic algorithms is 
$(\gamma+\lfloor\gamma\rfloor) / \lfloor\gamma\rfloor$,
which is always in the range $[2,3)$.  
This result was extended by Jurdzinski~\etal~\cite{JurdzinskiKL13},
who proved that the optimum ratio for the case of multiple (though fixed)
packet sizes is given by the same formula for the two packet sizes which
maximize it. 

Moreover, Jurdzinski~\etal~\cite{JurdzinskiKL13} gave further
results for \emph{divisible} packet sizes, i.e., instances in which every packet
size divides every larger packet size.  In particular,
they proved that on such instances speed $2$ is sufficient for $1$-competitiveness
in the resource augmentation setting.
(Note that the above formula for the optimal competitive ratio without speedup
gives $2$ for divisible instances.)

In another work, Anta \etal~\cite{AntaGKZ15} consider popular scheduling algorithms
and analyze their performance under speed augmentation with respect to three efficiency
measures, which they call \emph{completed load}, \emph{pending load}, and \emph{latency}.
The first is precisely the objective that we aim to maximize,
the second is the total size of the available but not yet completed packets
(which we minimize in turn), and finally, the last one is
the maximum time elapsed from a packet's arrival till the end of its successful transmission.
We note that a $1$-competitive algorithm (possibly with an additive constant) for any 
of the first two objectives is also $1$-competitive for the other, but there is no
similar relation for larger ratios. 

We note that Anta \etal~\cite{Anta13} demonstrate the necessity of
instantaneous error feedback by proving that discovering errors upon completed
transmission rules out a constant competitive ratio. They also provide
improved results for a stochastic online setting.

Recently, Kowalski, Wong, and Zavou~\cite{kowalski_fault_tolerant_2_sizes_17}
studied the effect of speedup on latency and pending load objectives
in the case of two packet sizes only. They use two conditions on the speedup, defined in~\cite{Anta13-dual},
and show that if both hold, then there is no
$1$-competitive deterministic algorithm for either objective (but speedup must be below $2$),
while if one of the conditions is not satisfied, such an algorithm exists.

\subsubsection{Multiple Channels or Machines}

The problem we study was generalized to multiple communication channels, 
machines, or processors, depending on particular application.  The standard
assumption, in communication parlance, is that the jamming errors on each
channel are independent, and that any packet can be transmitted on at most
one channel at any time.

For divisible instances, Jurdzinski~\etal~\cite{JurdzinskiKL13} extended their
(optimal) $2$-compe\-ti\-ti\-ve algorithm to an arbitrary number $f$ of channels.
The same setting, without the restriction to divisible instances was studied
by Anta \etal~\cite{Anta13-dual}, who consider the objectives of minimizing
the number or the total size of pending (i.e., not yet transmitted) packets.
They investigate what speedup is necessary and sufficient for $1$-competitiveness
with respect to either objective.  Recall that $1$-competitiveness for minimizing
the total size of pending packets is equivalent to $1$-competitiveness for our objective of maximizing the total
size of completed packets.  In particular, for either objective, Anta \etal~\cite{Anta13-dual}
obtain a tight bound of 2 on speedup for 1-competitiveness for two packet sizes.
Moreover, they claim a $1$-competitive algorithm with speedup $7/2$
for a constant number of sizes and pending (or completed) load,
but the proof is incorrect; see Section~\ref{sec:examples} for a (single-channel)
counterexample. 

Georgio \etal~\cite{GeorgiouK15} consider the same problem in a distributed
setting, distinguishing between different information models.  As communication
and synchronization pose new challenges, they restrict their attention to
jobs of unit size only and no speedup.  On top of efficiency measured by
the number of pending jobs, they also consider the standard (in distributed systems)
notions of correctness and fairness.

Finally, Garncarek~\etal~\cite{GarncarekJL17} consider ``synchronized'' parallel
channels that all suffer errors at the same time.
Their work distinguishes between ``regular'' jamming errors and 
``crashes'', which also cause the algorithm's state to reset, losing
any information stored about the past events.  They proved that for two packet sizes
the optimum ratio tends to $4/3$ for the former and to
$\phi = (\sqrt{5}+1)/2 \approx 1.618$ for the latter setting as the number $f$ of 
channels tends to infinity.

\subsubsection{Randomization}

All aforementioned results, as well as our work, concern deterministic algorithms.
In general, randomization often allows an improved competitive ratio, but while
the idea is simply to replace algorithm's cost or profit with its expectation
in the competitive ratio, the latter's proper definition is subtle:
One may consider the adversary's ``strategies'' for creating
and solving an instance separately, possibly limiting their powers.
Formal considerations lead to more than one \emph{adversary model},
which may be confusing.  As a case in point, Anta \etal~\cite{Anta13} note that
their lower bound strategy for two sizes (in our model) applies to randomized algorithms as well,
which would imply that randomization provides no advantage. 
However, their argument only applies to the adaptive adversary model,
which means that in the lower bound strategy
the adversary needs to make decisions based on the previous behavior of the algorithm
that depends on random bits.
To our best knowledge, randomized algorithms for our problem were never considered
for the more common oblivious adversary model, where the adversary needs to fix
the instance in advance and cannot change it according to the decisions
of the algorithm.  For more details and formal definitions of these
adversary models, refer to the article that first distinguished them~\cite{Ben-DavidBKTW94}
or the textbook on online algorithms~\cite{BEY}.

\subsection{Our Results}

Our major contribution is a uniform algorithm, called \emph{PrudentGreedy} (\MAIN),
described in Section~\ref{sec:alg}, that works well in every setting,
together with a uniform analysis framework (in Section~\ref{sec:local}).
This contrasts with the results of Jurdzinski~\etal~\cite{JurdzinskiKL13},
where each upper bound was attained by a dedicated algorithm with independently
crafted analysis; in a sense, this means that their algorithms require
the knowledge of speed they are running at.  Moreover, algorithms in~\cite{JurdzinskiKL13} do
require the knowledge of all admissible packet sizes.  Our algorithm
has the advantage that it is completely oblivious, i.e., requires no
such knowledge.  Furthermore, our algorithm is more appealing as it is
significantly simpler and ``work-conserving'' or ``busy'', i.e., transmitting
some packet whenever there is one pending, which is desirable in practice.
In contrast, algorithms in~\cite{JurdzinskiKL13} can be unnecessarily idle
if there is a small number of pending packets.

Our main result concerns the analysis of the general case
with speedup where we show that speed $4$ is sufficient for
our algorithm \MAIN to be $1$-competitive; the proof is by a 
complex (non-local) charging argument described in Section~\ref{sec:four}.

However, we start by formulating a simpler local analysis framework in Section~\ref{sec:local},
which is very universal as we demonstrate by applying it to several settings.
In particular, we prove that on general instances \MAIN achieves
the optimal competitive ratio of~$3$ without speedup and we also get a trade-off
between the competitive ratio and the speedup for speeds in $[1,4)$;
see Figure~\ref{fig:speedVsCR} for a graph of our bounds
on the competitive ratio depending on the speedup.

To recover the $1$-competitiveness at speed $2$ and also $2$-competitiveness at speed $1$
for divisible instances, we have to modify our algorithm slightly as otherwise,
we can guarantee $1$-competitiveness for divisible instances
only at speed $2.5$ (Section~\ref{sec:main-divisible}).  This is to be
expected as divisible instances are a very special case.
The definition of the modified algorithm for divisible instances and its analysis
by our local analysis framework is in Section~\ref{sec:div}.

On the other hand, we prove that our original algorithm is $1$-competitive on far broader
class of ``well-separated'' instances at sufficient speed: If the ratio between 
two successive packet sizes (in their sorted list) is no smaller than
$\alpha \geq 1$, our algorithm is $1$-competitive if its speed is at least
$S_\alpha$ which is a non-increasing function of $\alpha$ such that $S_1=6$ and
$\lim_{\alpha \to \infty} S_\alpha = 2$; see Section~\ref{sec:wellSeparated}
for the precise definition of $S_\alpha$.
(Note that speed $4$ is sufficient for $1$-competitiveness, but having $S_1=6$
reflects the limits of the local analysis.)

In Section~\ref{sec:examples} we demonstrate that
the analyses of our algorithm are mostly tight, i.e., that (a) on general 
instances, the algorithm is no better than $(1+2/s)$-competitive for $s < 2$
and no better than $4/s$-competitive for $s\in [2,4)$,
(b) on divisible instances, it is no better than $4/3$-competitive for 
$s<2.5$, and (c) it is at least $2$-competitive for $s < 2$, even for
two divisible packet sizes (example (c) is in Section~\ref{sec:example-2sizes}).
See Figure~\ref{fig:speedVsCR} for a graph of our bounds.
Note that we do not obtain tight bounds for $s\in [2,4)$,
but we conjecture that using an appropriately adjusted non-local
analysis of Theorem~\ref{thm:four} (which shows $1$-competitiveness for $s=4$),
it is possible to show that the algorithm is $4/s$-competitive for $s\in [2,4)$.

\begin{figure}
	\begin{center}
		\input{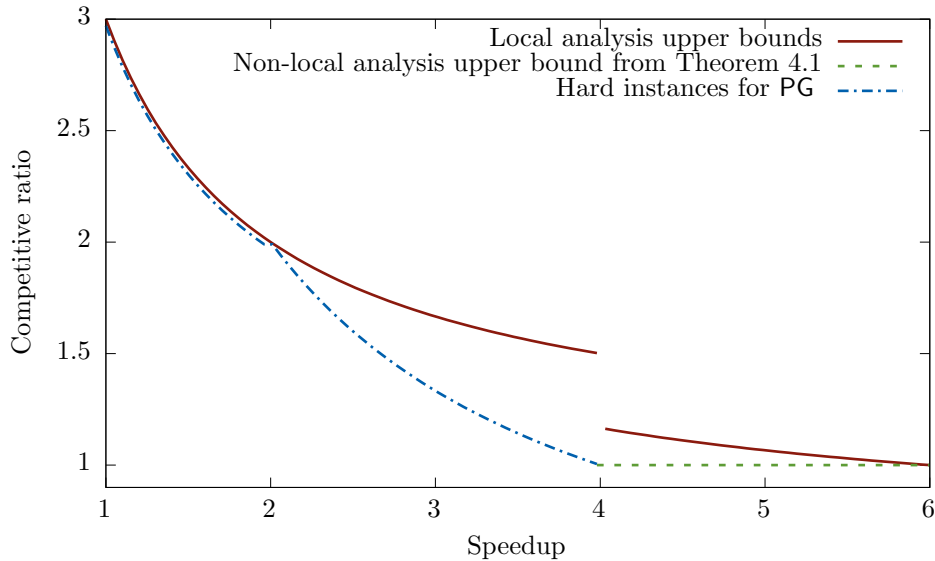}
	\end{center}
	\caption{A graph of our upper and lower bounds on the competitive ratio of algorithm $\MAINs$,
		depending on the speedup $s$.
		The upper bounds are from Theorems~\ref{thm:gen} and~\ref{thm:four} and the lower
		bounds are by hard instances from Section~\ref{sec:examples}.
	} 
	\label{fig:speedVsCR}
\end{figure}

In Section~\ref{sec:LBs} we complement these results with two lower bounds on the speed that is sufficient
to achieve $1$-competitiveness by a deterministic algorithm.
The first one proves that even for two divisible packet sizes, speed $2$ is 
required to attain $1$-competitiveness, establishing optimality of our modified
algorithm and that of Jurdzinski~\etal~\cite{JurdzinskiKL13} for the divisible case.
The second lower bound strengthens the previous construction by showing that for non-divisible
instances with more packet sizes, speed $\phi+1 \approx 2.618$ is needed
for $1$-competitiveness.  Both results hold even if all packets are released
simultaneously.

We remark that Sections~\ref{sec:local}, \ref{sec:four}, and~\ref{sec:LBs}
are independent on each other and can be read in any order.
In particular, the reader may safely skip proofs for 
special instances in Section~\ref{sec:local} (e.g., the divisible instances),
and proceed to Section~\ref{sec:four} with the main result,
which is $1$-competitiveness with speedup $4$.

\section{Algorithms, Preliminaries, Notations}\label{sec:prelims}

We start by some notations.  We assume there are $k$ distinct non-zero
packet sizes denoted by $\ell_i$ and ordered so that
$\ell_1<\cdots<\ell_k$. For convenience, we define $\ell_0=0$.  We say
that the packet sizes are divisible if $\ell_i$ divides $\ell_{i+1}$
for all $i=1,\ldots,k-1$. For a packet $p$, let $\ell(p)$ denote the
size of $p$.  For a set of packets $P$, let $\ell(P)$ denote the total
size of all the packets in $P$.

During the run of an algorithm, at time $t$, a packet is pending if it
is released before or at $t$, not completed before or at $t$ and not started before $t$ and still running. At time $t$,
if no packet is running, the algorithm may start any pending packet.
As a convention of our model, if a fault (jamming error) happens at
time $t$ and this is the completion time of a previously scheduled
packet, this packet is considered completed. Also, at the fault time,
the algorithm may start a new packet.

Let $L_{\ALG}(i,Y)$ denote the total size of packets of size $\ell_i$
completed by an algorithm \ALG\ during a time interval $Y$.
Similarly, $L_{\ALG}(\unarygeq i,Y)$ (resp. $L_{\ALG}(\unarylt i,Y)$) denotes the
total size of packets of size at least $\ell_i$ (resp. less than
$\ell_i$) completed by an algorithm \ALG\ during a time interval $Y$;
formally we define $L_{\ALG}(\unarygeq i,Y)=\sum_{j=i}^k L_{\ALG}(j,Y)$ and
$L_{\ALG}(\unarylt i,Y)=\sum_{j=1}^{i-1} L_{\ALG}(j,Y)$. We use the notation
$L_{\ALG}(Y)$ with a single parameter to denote the size
$L_{\ALG}(\unarygeq 1,Y)$ of packets of all sizes completed by \ALG during
$Y$ and the notation $L_\ALG$ without parameters to denote the size of
all packets of all sizes completed by \ALG at any time.

By convention, the schedule starts at time $0$ and ends at time $T$,
which is a part of the instance unknown to an online algorithm until
it is reached. (This is similar to the times of jamming errors, one
can also alternatively say that after $T$ the errors are so frequent
that no packet is completed.) Algorithm \ALG is called
$R$-competitive, if there exists a constant $A$, possibly dependent on
$k$ and $\ell_1$, \ldots, $\ell_k$, such that for any instance and its
optimal schedule $\OPT$ we have
$L_{\OPT}\leq R \cdot L_\ALG+A$.
We remark that in our analyses we show only a crude bound on $A$.

We denote the algorithm {\ALG} with speedup $s\geq 1$ by $\ALGs$.  The
meaning is that in $ALG(s)$, packets of size $L$ need time $L/s$ to
process. In the resource-augmentation variant, we are mainly interested in
finding the smallest $s$ such that $ALG(s)$ is $1$-competitive, compared
to $\OPT=\OPT(1)$ that runs at speed $1$. 

\subsection{Algorithm PrudentGreedy (\MAIN)}\label{sec:alg}

The general idea of the algorithm is that after each error, we start
by transmitting packets of small sizes, only increasing the size of
packets after a sufficiently long period of uninterrupted
transmissions. It turns out that the right tradeoff is to transmit a
packet only if it would have been transmitted successfully if started
just after the last error. It is also crucial that the initial packet
after each error has the right size, namely to ignore small
packet sizes if the total size of remaining packets of those
sizes is small compared to a larger packet that can be
transmitted. In other words, the size of the first transmitted
packet is larger than the \textit{total} size of all pending smaller packets
and we choose the largest such size.
This guarantees that if no error occurs, all currently pending packets with
size equal to or larger than the size of the initial packet are
eventually transmitted before the algorithm starts a smaller packet.

We now give the description of our algorithm \emph{PrudentGreedy} (\MAIN) for general packet
sizes, noting that the other algorithm for divisible sizes differs only
slightly.  We divide the run of the algorithm into phases.
Each phase starts by an invocation of
the initial step in which we need to carefully select a packet to
transmit as discussed above. The phase ends by a fault, or when there
is no pending packet, or when there are pending packets only of sizes
larger than the total size of packets completed in the current
phase.  The periods of idle time, when no packet is pending,
do not belong to any phase.

Formally, throughout the algorithm, $t$ denotes the current time. The
time $t_B$ denotes the start of the current phase; initially $t_B =
0$.  We set $\rel(t)=s\cdot (t-t_B)$. Since the algorithm does not
insert unnecessary idle time, $\rel(t)$ denotes the amount of
transmitted packets in the current phase. Note that we use $\rel(t)$
only when there is no packet running at time $t$, so there is no
partially executed packet. Thus $\rel(t)$ can be thought of as a
measure of time relative to the start of the current phase (scaled by
the speed of the algorithm). Note also that the algorithm can evaluate
$\rel(t)$ without knowing the speedup, as it can simply observe the
total size of the transmitted packets.  Let $P^{<i}$ denote the set of
pending packets of sizes $\ell_1$, \ldots, $\ell_{i-1}$ at any given
time.


\begin{center}
\fbox{\parbox{0.9\textwidth}{
{\bf Algorithm PrudentGreedy (\MAIN)}           
\begin{enumerate}[label=(\arabic*), nosep]
\item 
If no packet is pending, stay idle until the next release time.
\item
Let $i$ be the maximal $i \leq k$ such that there is a pending packet
of size $\ell_i$ and $\ell(P^{<i})<\ell_i$.  Schedule a packet of size
$\ell_i$ and set $t_B = t$.
\item
  Choose the maximum $i$ such that
\\
\mbox{}
\quad
(i) there is a pending packet of size $\ell_i$, 
\\
\mbox{}
\quad
(ii) $\ell_i \leq \rel(t)$.
\\
Schedule a packet of size $\ell_i$.  Repeat Step (3) as long as such $i$
exists.
\item
If no packet satisfies the condition in Step (3), go to Step (1). 
\end{enumerate}
}}
\end{center}

We first note that the algorithm is well-defined, i.e., that it is always
able to choose a packet in Step~(2) if it has any packets pending, and
that if it succeeds in sending it, the length of thus started phase can
be related to the total size of the packets completed in it.

\begin{lemma}\label{lem:simple} 
In Step~(2), \MAIN always chooses some packet if it has any pending.
Moreover, if \MAIN completes the first packet in the phase, then
$L_\MAINs((t_B,t_E])>s\cdot (t_E-t_B)/2$, where $t_B$ denotes the start of the phase
and $t_E$ its end (by a fault or Step~(4)).
\end{lemma}

\begin{proof}
For the first property, note that a pending packet of the smallest size is
eligible.
For the second property, note that there is no idle time in the phase,
and that only the last packet chosen by \MAIN in the phase may not complete
due to a jam.  By the condition in Step~(3), the size of this jammed packet is
no larger than the total size of all the packets \MAIN previously completed
in this phase (including the first packet chosen in Step~(2)), which yields the bound.
\end{proof}

The following lemma shows a crucial property of the algorithm, namely
that if packets of size $\ell_i$ are pending, the algorithm schedules
packets of size at least $\ell_i$ most of the time. Its proof also
explains the reasons behind our choice of the first packet in a phase
in Step (2) of the algorithm.
%

\begin{lemma}
\label{l:main}
Let $u$ be a start of a phase in $\MAINs$ and $t=u+\ell_i/s$.
\begin{enumerate}[label=\rm(\roman*), nosep]
\item
If a packet of size $\ell_i$ is pending at time $u$ and no fault
occurs in $(u,t)$, then the phase does not end before $t$.
\item
Suppose that $v>u$ is such that any time in $[u,v)$ a packet of size
$\ell_i$ is pending and no fault occurs. Then the phase 
does not end in $(u,v)$ and $L_\MAINs(<i,(u,v])<
\ell_i+\ell_{i-1}$. (Recall that $\ell_0=0$.)
\end{enumerate}
\end{lemma}
\begin{proof}
(i) Suppose for a contradiction that the phase started at $u$ ends at
time $t'<t$. We have $\rel(t')<\rel(t)=\ell_i$.  Let $\ell_j$ be the
smallest packet size among the packets pending at $t'$. As there is
no fault, the reason for a new phase has to be that
$\rel(t')<\ell_j$, and thus Step (3) did not choose a packet to be
scheduled. Also note that any packet started before $t'$ was
completed. This implies, first, that there is a pending packet of
size $\ell_i$, as there was one at time $u$ and there was
insufficient time to complete it, so $j$ is well-defined and $j\le i$.  Second, all
packets of sizes smaller than $\ell_j$ pending at $u$ were completed
before $t'$, so their total size is at most
$\rel(t')<\ell_j$. However, this contradicts the fact that the phase
started by a packet smaller than $\ell_j$ at time $u$, as a pending packet
of the smallest size equal to or larger than $\ell_j$ satisfied the
condition in Step (2) at time $u$ and a packet of size $\ell_i$ is
pending at $u$. (Note that it is possible that no packet of
size $\ell_j$ is pending at $u$.)

(ii) By (i), the phase that started at $u$ does not end before time $t$ if
no fault happens. A packet of size $\ell_i$ is always pending by the
assumption of the lemma, and it is always a valid choice of a packet
in Step (3) from time $t$ on. Thus, the phase that started at $u$
does not end in $(u,v)$, and moreover only packets of sizes at least
$\ell_i$ are started in $[t,v)$. It follows that packets of sizes
smaller than $\ell_i$ are started only before time $t$ and their
total size is thus less than
$\rel(t)+\ell_{i-1}=\ell_i+\ell_{i-1}$. The lemma follows.
\end{proof}

\section{Local Analysis and Results}
\label{sec:local}

In this section we formulate a general method for analyzing our
algorithms by comparing locally within each phase the size of
``large'' packets completed by the algorithm and by the
adversary. This method simplifies a complicated induction used 
in~\cite{JurdzinskiKL13}, letting us obtain the same upper bounds of $2$ and $3$ 
on competitiveness for divisible and unrestricted packet sizes, respectively, at 
no speedup, as well as several new results for the non-divisible
cases included in this section.  In Section~\ref{sec:four},
we use a more complex charging scheme to obtain our main result.
We postpone the use of local analysis for the divisible case to
Section~\ref{sec:div}.

For the analysis, 
let $s\geq 1$ be the speedup. We fix an instance and its
schedules for $\MAINs$ and \OPT.

\subsection{Critical Times and Master Theorem}
\label{sec:master}

The common scheme is the following. We carefully
define a sequence of critical times $C_k\leq C_{k-1}\leq
\cdots \leq C_1\leq C_0$, where $C_0=T$ is the end of the schedule,
satisfying two properties: (1) till time $C_i$ the algorithm has completed
almost all pending packets of size $\ell_i$
released before $C_i$ and (2) in $(C_i,C_{i-1}]$, a packet of size $\ell_i$ is always
pending. Properties~(1) and~(2) allow us to relate $L_\OPT(i,(0,C_i])$
and $L_\OPT(\geq i,(C_i,C_{i-1}])$, respectively, to their ``\MAIN\ counterparts''.
As each packet completed by \OPT\ belongs to exactly one of these sets,
summing the bounds gives the desired results; see Figure~\ref{fig:local} for an illustration.
These two facts together imply
$R$-competitiveness of the algorithm for appropriate $R$ and speed
$s$.

We first define the notion of $i$-good times so that they satisfy
property~(1), and then choose the critical times among their suprema so that
those satisfy property~(2) as well.
\begin{definition}
  \label{def:critical}
Let $s\geq 1$
be the speedup. For $i=1,\ldots k$, time $t$ is called $i$-good if one
of the following conditions holds:
\begin{enumerate}[label=(\roman*), nosep]
\item
At time $t$, algorithm $\MAINs$ starts a new phase by scheduling a packet of size
larger than $\ell_i$, or
\item
at time $t$, no packet of size $\ell_i$ is pending for $\MAINs$, or
\item
$t=0$.
\end{enumerate}
\noindent
We define critical times $C_0, C_1,\ldots,C_k$ iteratively as follows:
\begin{itemize}[nosep]
\item
$C_0=T$, i.e., it is the end of the schedule. 
\item
For $i=1,\ldots,k$, $C_i$ is the supremum of $i$-good times $t$ such that
$t\leq C_{i-1}$.
\end{itemize}
\end{definition}

Note that all $C_i$'s are defined and $C_i\geq0$, as time $t=0$ is
$i$-good for all $i$.
The choice of $C_i$ implies that each $C_i$ is of one of the three
types (the types are not disjoint):
\begin{itemize}[nosep]
\item
  $C_i$ is $i$-good and a phase starts at $C_i$ (this includes $C_i=0$),
\item
  $C_i$ is $i$-good and $C_i=C_{i-1}$, or
\item
  there exists a packet of size $\ell_i$
  pending at $C_i$, however, any such packet was released at $C_i$.
\end{itemize}
If the first two options do not apply, then the last one 
is the only remaining possibility (as otherwise some time
in the non-empty interval $(C_i,C_{i-1}]$ would be $i$-good); in this case,
$C_i$ is not $i$-good, but only the supremum of $i$-good times.

First we bound the total size of packets of size $\ell_i$ completed before
$C_i$; the proof actually only uses the fact that each $C_i$ is the
supremum of $i$-good times and justifies the definition above.
\begin{lemma}
\label{l:master-small}
Let $s\geq 1$ be the speedup. Then, for any $i$, it holds
$L_{\OPT}(i,(0,C_i])\leq L_{\MAINs}(i,(0,C_i])+\ell_k$.
\end{lemma}
\begin{proof}
If $C_i$ is $i$-good and satisfies condition (i) in
Definition~\ref{def:critical}, then by the description of Step (2) of
the algorithm, the total size of pending packets of size $\ell_i$ is
less than the size of the scheduled packet, which is at most $\ell_k$
and the lemma follows.

In all the remaining cases it holds that \MAINs has completed all the
jobs of size $\ell_i$ released before $C_i$, thus the inequality holds
trivially even without the additive term.
\end{proof}

Our remaining goal is to bound $L_\OPT(\geq i,(C_i,C_{i-1}])$. 
We divide $(C_i,C_{i-1}]$ into $i$-segments by the faults. We prove
the bounds separately for each $i$-segment. One important fact is that
for the first $i$-segment we use only a loose bound, as we can use the
additive constant. The critical part is then the bound for
$i$-segments started by a fault, this part determines the competitive ratio and is
different for each case. We summarize the general method by the
following definition and master theorem.

\begin{definition}
The interval $(u,v]$ is called the initial $i$-segment if
$u=C_i$ and $v$ is either $C_{i-1}$ or the first time of a fault after
$u$, whichever comes first. 

The interval $(u,v]$ is called a proper $i$-segment if
$u\in(C_i,C_{i-1})$ is a time of a fault  and $v$ is either $C_{i-1}$
or the first time of a fault after $u$, whichever comes first. 
\end{definition}
Note that there is no $i$-segment if $C_{i-1}=C_i$.

\begin{theorem}[Master Theorem]\label{thm:master}
Let $s\geq 1$ be the speedup. Suppose that for $R\geq 1$ both of the following hold:
\begin{enumerate}
	\item For each $i=1,\ldots,k$ and each proper $i$-segment $(u,v]$
	with $v-u\geq \ell_i$, it holds that
	\begin{equation}
	\label{eq:master}
	(R-1)L_{\MAINs}((u,v])+L_{\MAINs}(\geq i,(u,v])
	\;\geq\;
	L_\OPT(\geq i,(u,v])
	\,.
	\end{equation}
	\item For the initial $i$-segment $(u,v]$, it holds that
	\begin{equation}
	\label{eq:master-init}
	L_\MAINs(\geq i,(u,v])
	\;>\;
	s(v-u)-4\ell_k
	\,.
	\end{equation}
\end{enumerate}
Then $\MAINs$ is $R$-competitive.
\end{theorem}

\begin{proof}
First note that for a proper $i$-segment $(u,v]$, $u$ is a fault
time. Thus if $v-u<\ell_i$, then $L_\OPT(\geq i,(u,v])=0$ and
(\ref{eq:master}) is trivial. It follows that (\ref{eq:master})
holds even without the assumption $v-u\geq \ell_i$.

Now consider the initial $i$-segment $(u,v]$. We have $L_\OPT(\geq
i,(u,v])\leq\ell_k+ v-u$, as at most a single packet started before
$u$ can be completed.  Combining this with \eqref{eq:master-init} and using $s\ge 1$, we get
$L_\MAINs(\geq i,(u,v])\;>\;s(v-u)-4 \ell_k\;\geq\;v-u-4\ell_k\;\geq\;L_\OPT(\geq i,(u,v])-5\ell_k$.

Summing this with \eqref{eq:master} for all proper $i$-segments and
using $R\geq 1$ we get
\begin{align}
(R-1)L_\MAINs((C_i,C_{i-1}])+ L_\MAINs(\geq i,(C_i,C_{i-1}])+5\ell_k \nonumber \\
 \geq L_\OPT(\geq i,(C_i,C_{i-1}])\,. \label{eq:master-phase}
\end{align}
Note that for $C_i=C_{i-1}$, Equation \eqref{eq:master-phase}
holds trivially.

To complete the proof, note that each completed job in
the optimum contributes to exactly one among the $2k$ terms
$L_\OPT(\geq i,(C_i,C_{i-1}])$ and $L_\OPT(i,(0,C_i])$; similarly for
$L_\MAINs$. Thus by summing both \eqref{eq:master-phase} and
Lemma~\ref{l:master-small} for all $i=1,\ldots,k$ we obtain
\begin{eqnarray*}
L_\OPT
&=& 
\sum_{i=1}^kL_\OPT\left(\geq i,(C_i,C_{i-1}]\right)+\sum_{i=1}^kL_\OPT(i,(0,C_i])
\\
&\leq& 
\sum_{i=1}^k \bigg( (R-1)L_\MAINs((C_i,C_{i-1}]) +
\left(L_\MAINs(\geq i,(C_i,C_{i-1}])+5\ell_k \right)\bigg)
\\
& & \mbox{}
+
\sum_{i=1}^k\left(L_\MAINs(i,(0,C_i])+\ell_k\right)
\\
&\leq&
(R-1)L_\MAINs+L_\MAINs+6k\ell_k
\;=\;R\cdot L_\MAINs+6k\ell_k
\,. 
\end{eqnarray*}
The theorem follows.
\end{proof}

\begin{figure}[!ht]
  \begin{center}
    \includegraphics{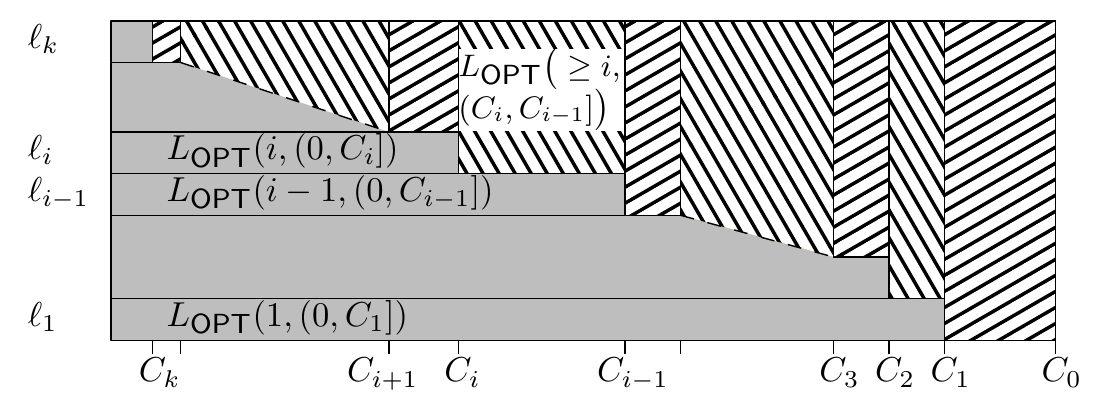}
  \end{center}
  \caption{An illustration of dividing the schedule of \OPT in the local analysis, i.e.,
  dividing the (total size of) packets completed by \OPT into $L_\OPT(i,(0,C_i])$
  and $L_\OPT\left(\geq i,(C_i,C_{i-1}]\right)$ for $i=1,\dots, k$.
  Rows correspond to packet sizes and the X-axis to time.
  Gray horizontal rectangles thus correspond to $L_\OPT(i,(0,C_i])$, i.e.,
  these rectangles represent the time interval $(0,C_i]$ and packets of size $\ell_i$ completed by \OPT 
  in $(0,C_i]$, whereas hatched rectangles correspond to $L_\OPT\left(\geq i,(C_i,C_{i-1}]\right)$.
  }
  \label{fig:local}
\end{figure}

\subsection{Local Analysis of PrudentGreedy (\MAIN)}
\label{sec:local-main}

The first part of the following lemma implies the condition
\eqref{eq:master-init} for the initial $i$-segments in all cases. The
second part of the lemma is the base of the analysis of a proper
$i$-segment, which is different in each situation.
\begin{lemma}
\label{l:main-segments}
\begin{enumerate}[label=\rm(\roman*), nosep]
\item
If $(u,v]$ is the initial $i$-segment, 
then $L_\MAINs(\geq i,(u,v])> s(v-u)-4\ell_k$.
\item
If $(u,v]$ is a proper $i$-segment and $v-u\geq \ell_i$ then
$L_\MAINs((u,v])>s(v-u)/2$ and
$L_\MAINs(\geq i,(u,v])>s(v-u)/2-\ell_i-\ell_{i-1}$. (Recall that $\ell_0=0$.)  
\end{enumerate}
\end{lemma}
\begin{proof}
(i)
If the phase that starts at $u$
or contains $u$ ends before $v$, let $u'$ be its end; otherwise let
$u'=u$. We have $u'\leq u+\ell_i/s$, as otherwise any packet of size
$\ell_i$, pending throughout the $i$-segment by definition, would be an
eligible choice in Step (3) of the algorithm, and the phase would not end
before $v$. Using Lemma~\ref{l:main}(ii), we have
$L_\MAINs(<i,(u',v])<\ell_i+\ell_{i-1}<2\ell_k$. Since at most one
packet at the end of the segment is unfinished, we have
$L_\MAINs(\geq i,(u,v])
\geq
L_\MAINs(\geq i,(u',v])
> s(v-u')-3\ell_k
\ge s(v-u)-4\ell_k$.

(ii)
Let $(u,v]$ be a proper $i$-segment. Thus $u$ is a start of a phase
that contains at least the whole interval $(u,v]$ by Lemma~\ref{l:main}(ii).
By the definition of $C_i$, $u$ is not $i$-good,
so the phase starts by a packet of size at most $\ell_i$.
If $v-u\geq \ell_i$ then the first packet finishes (as $s\geq 1$) and thus
$L_\MAINs((u,v])>s(v-u)/2$ by Lemma~\ref{lem:simple}.
The total size of completed packets smaller than
$\ell_i$ is less than $\ell_i+\ell_{i-1}$ by Lemma~\ref{l:main}(ii), and
thus $L_\MAINs(\geq i,(u,v])>s(v-u)/2-\ell_i-\ell_{i-1}$.
\end{proof}

\subsubsection{General Packet Sizes}

The next theorem gives a tradeoff of the competitive ratio of $\MAINs$
and the speedup $s$ using our local analysis. While Theorem~\ref{thm:four} shows that $\MAINs$ is
$1$-competitive for $s\geq 4$, here we give a weaker result that
reflects the limits of the local analysis.  However, for
$s=1$ our local analysis is tight as already the lower bound
from~\cite{Anta13} shows that no algorithm is better than 3-competitive 
(for packet sizes $1$ and $2-\eps$).  
See Figure~\ref{fig:speedVsCR} for an illustration of our upper
and lower bounds on the competitive ratio of $\MAINs$.
\begin{theorem}
\label{thm:gen}
$\MAINs$ is $R_s$-competitive where:

$R_s=1+2/s$ for $s\in[1,4)$,

$R_s=2/3+2/s$ for $s\in[4,6)$, and

$R_s=1$ for $s\geq 6$.
\end{theorem}
\begin{proof}
Lemma~\ref{l:main-segments}(i) implies the condition
\eqref{eq:master-init} for the initial $i$-segments.  We now
prove (\ref{eq:master}) for any proper $i$-segment $(u,v]$ with
$v-u\geq\ell_i$ and appropriate $R$. The bound then follows by the
Master Theorem.

Since there is a fault at time $u$, we have $L_\OPT(\geq i,(u,v])\leq v-u$.

\smallskip
For $s\geq 6$, Lemma~\ref{l:main-segments}(ii) implies
\begin{align*}
L_\MAINs(\geq i,(u,v])&>s(v-u)/2-2\ell_i\\
&\geq 3(v-u)-2(v-u)=v-u\geq L_\OPT(\geq i,(u,v])\,,
\end{align*}
which is (\ref{eq:master}) for $R=1$.

\smallskip
For $s\in[4,6)$, by Lemma~\ref{l:main-segments}(ii) we have 
$L_\MAINs((u,v]) > s(v-u)/2$ and by multiplying it by $(2/s - 1/3)$
we obtain
\[
\left(\frac{2}{s}-\frac{1}{3}\right)\cdot L_\MAINs((u,v])>
\left(1-\frac{s}{6}\right)(v-u)
\,.
\]
Thus to prove (\ref{eq:master}) for $R=2/3+2/s$, it suffices to show that
\[
L_\MAINs(\geq i,(u,v])>\frac{s}{6}(v-u)
\,,
\]
as clearly $v-u\geq L_\OPT(\geq i,(u,v])$.  The remaining inequality again follows 
from Lemma~\ref{l:main-segments}(ii), but we need to consider two cases:

If $(v-u)\geq\frac{6}{s}\ell_i$, then 
\[
L_\MAINs(\geq i,(u,v])>\frac{s}{2}(v-u)-2\ell_i
\geq
\frac{s}{2}(v-u) - \frac{s}{3}(v-u)
= \frac{s}{6}(v-u)
\,.
\]
On the other hand, if $(v-u)<\frac{6}{s}\ell_i$, then using $s\geq 4$ as well,
\[
L_\MAINs(\geq i,(u,v])>\frac{s}{2}(v-u)-2\ell_i\geq 0\,,
\]
therefore $\MAINs$ completes a packet of size at least $\ell_i$ which implies
\[
L_\MAINs(\geq i,(u,v]) \geq \ell_i > \frac{s}{6}(v-u)\,,
\]
concluding the case of $s\in[4,6)$.

\smallskip
For $s\in[1,4)$, by Lemma~\ref{l:main-segments}(ii) we get
$(2/s)\cdot L_\MAINs((u,v])>v-u\geq L_\OPT(\geq i,(u,v])$,
which implies (\ref{eq:master}) for $R=1+2/s$.
\end{proof}

\subsubsection{Well-separated Packet Sizes}\label{sec:wellSeparated}

We can obtain better bounds on the speedup necessary for
$1$-competitiveness if the packet sizes are sufficiently
different. Namely, we call the packet sizes $\ell_1,\ldots,\ell_k$
\emph{$\alpha$-separated} if $\ell_i\geq \alpha\ell_{i-1}$ holds for
$i=2,\ldots,k$. 

\begin{figure}
  \begin{center}
    \input{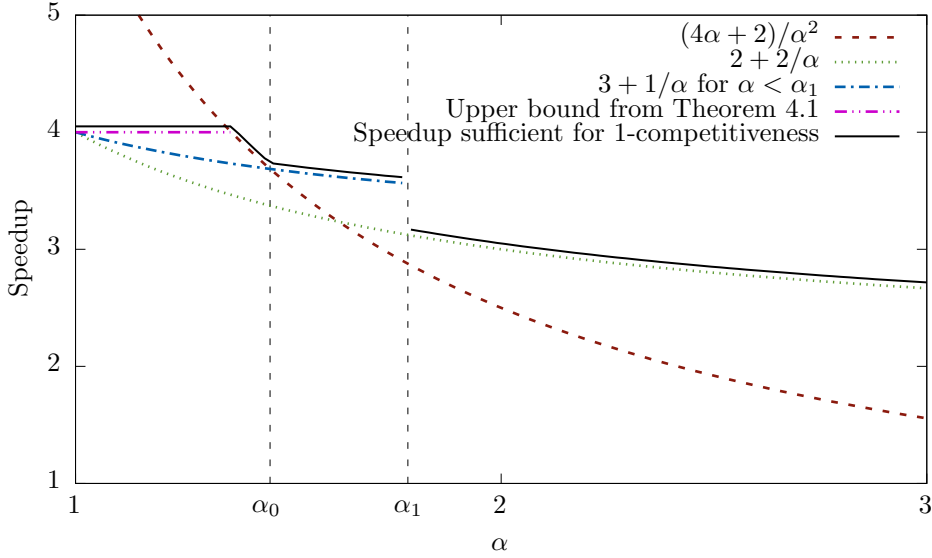}
  \end{center}
  \caption{A graph of $S_\alpha$ and the bounds on the speedup that we use
  in Theorem~\ref{thm:main-separated}.
  Note that in the graph we also use Theorem~\ref{thm:four} for $1$-competitiveness with speed $4$ (for any $\alpha$),
  but in the definition of $S_\alpha$ we do not take it into account.
  }
  \label{fig:speedVsSeparation}
\end{figure}

Next, we show that for $\alpha$-separated packet sizes,
$\MAIN(S_\alpha)$ is $1$-competitive for the following $S_\alpha$. We
define 
\begin{align*}
\alpha_0 &=\frac12+\frac16\sqrt{33}\approx 1.46
\mbox{\,, which is the positive root of $3\alpha^2-3\alpha-2$.}
\\[1ex]
\alpha_1 &=\frac{3+\sqrt{17}}{4}\approx 1.78
\mbox{\,, which is the positive root of $2\alpha^2-3\alpha-1$.}
\\[1ex]
S_\alpha &=
\begin{cases}
\displaystyle
\frac{4\alpha+2}{\alpha^2}& \mbox{ for $\alpha\in[1,\alpha_0]$,}
\\[1ex]
3+\frac1\alpha &\mbox{ for $\alpha\in[\alpha_0,\alpha_1)$, and}
\\[1ex]
2+\frac2\alpha &\mbox{ for $\alpha\geq\alpha_1$.}
\end{cases}
\end{align*}
See Figure~\ref{fig:speedVsSeparation} for a graph of $S_\alpha$
and all the bounds on it that we use.
The value of $\alpha_0$ is chosen as the point where 
$(4\alpha+2)/\alpha^2=3+1/\alpha$. The value of $\alpha_1$ is chosen
as the point from which the argument in case (viii) of the proof below works, 
which allows for a better result for $\alpha\geq\alpha_1$.
If $s\geq S_\alpha$ then $s\geq (4\alpha+2)/\alpha^2$
and $s\geq 2+2/\alpha$ for all $\alpha$ and also $s\geq 3+1/\alpha$
for $\alpha < \alpha_1$; these facts follow from 
inspection of the functions and are useful for the analysis.

Note that $S_\alpha$ is decreasing in $\alpha$, with a single
discontinuity at $\alpha_1$. We have $S_1=6$, matching
the upper bound for 1-competitiveness using local analysis.  We have $S_2=3$, i.e., $\MAIN(3)$ is
$1$-competitive for $2$-separated packet sizes, which includes the
case of divisible packet sizes (however, only $s\ge 2.5$ is needed
in the divisible case, as we show later). The limit of $S_\alpha$
for $\alpha\rightarrow +\infty$ is $2$.  For
$\alpha<(1+\sqrt3)/2\approx 1.366$, we get $S_\alpha>4$, while
Theorem~\ref{thm:four} shows that $\MAINs$ is $1$-competitive for
$s\geq 4$; the weaker result of Theorem~\ref{thm:main-separated} below
reflect the limits of the local analysis.

\begin{theorem}
\label{thm:main-separated}
Let $\alpha>1$. If the packet sizes are $\alpha$-separated, then
$\MAINs$ is $1$-competitive for any $s\geq S_\alpha$.
\end{theorem}

\begin{proof}
Lemma~\ref{l:main-segments}(i) implies \eqref{eq:master-init}.  We now
prove for any proper $i$-segment $(u,v]$ with   $v-u\geq\ell_i$ that
\begin{equation}
\label{eq:master1}
L_\MAINs(\geq i,(u,v])\geq L_\OPT(\geq i,(u,v])
\,,
\end{equation}
which is \eqref{eq:master} for $R=1$. The bound then follows by the
Master Theorem.

Let $X=L_\OPT(\geq i,(u,v])$. Note that $X\leq v-u$. 

Lemma~\ref{l:main-segments}(ii) together with
$\ell_{i-1}\leq\ell_i/\alpha$ gives $L_\MAINs(\geq i,(u,v])>M$ for
$M=sX/2-(1+1/\alpha)\ell_i$.

We use the fact that both
$X$ and $L_\MAINs(\geq i,(u,v])$ are sums of some packet sizes
$\ell_j$, $j\geq i$, and thus only some of the values are
possible. However, the situation is quite complicated, as for example
$\ell_{i+1}$, $\ell_{i+2}$, $2\ell_i$, $\ell_i+\ell_{i+1}$ are
possible values, but their ordering may vary.

We distinguish several cases based on $X$ and $\alpha$.  We note in
advance that the first five cases suffice for $\alpha<\alpha_1$; only
after completing the proof for $\alpha<\alpha_1$, we analyze the
additional cases needed for $\alpha \geq \alpha_1$.

\mycase{i}
$X=0$. Then \eqref{eq:master1} is trivial. 

\mycase{ii}
$X=\ell_i$. Using $s\geq 2+2/\alpha$, we obtain $M\geq (1+1/\alpha)\ell_i-(1+1/\alpha)\ell_i = 0$.
Thus $L_\MAINs(\geq i,(u,v])>M\geq 0$ which implies
$L_\MAINs(\geq i,(u,v])\geq \ell_i=X$ and \eqref{eq:master1} holds.

\mycase{iii}
$X=\ell_{i+1}$ and $\ell_{i+1}\leq 2\ell_i$. Using $s\geq (4\alpha+2)/\alpha^2$ and
$X=\ell_{i+1}\geq\alpha\ell_i$, we obtain
$$M\geq \frac{s\ell_{i+1}}{2} -\left(1+\frac1\alpha\right)\ell_i \ge \left(2+\frac1\alpha\right)\ell_i-\left(1+\frac1\alpha\right)\ell_i=\ell_i\,.$$
Thus $L_\MAINs(\geq i,(u,v])>\ell_i$ which together with
$\ell_{i+1}\leq 2\ell_i$ implies
$L_\MAINs(\geq i,(u,v])\geq \ell_{i+1}=X$ and \eqref{eq:master1} holds.

\mycase{iv}
$X\geq \alpha^2\ell_i$.
(Note that this includes all cases when a
packet of size at least  $\ell_{i+2}$ contributes to $X$.)
We first show that $s\geq
2(1+1/\alpha^2+1/\alpha^3)$ by straightforward calculations with the
golden ratio $\phi$:
\begin{itemize}[nosep]
\item 
If $\alpha\leq\phi$, we have 
$$s\geq \frac{4\alpha+2}{\alpha^2}=2\left(\frac2\alpha + \frac{1}{\alpha^2}\right)\geq
2\left(1+\frac{1}{\alpha^2}+\frac{1}{\alpha^3}\right)\,,$$
where we use $2/\alpha \geq 1+1/\alpha^3$ or equivalently $\alpha^3 + 1 - 2\alpha^2\le 0$,
which is true as
$$\alpha^3 + 1 - 2\alpha^2 = \alpha^3 - \alpha^2 + 1 - \alpha^2
= \alpha^2(\alpha - 1) - (\alpha + 1)(\alpha - 1) = (\alpha-1)(\alpha^2-\alpha-1)\le 0\,,$$
where the last inequality holds for $\alpha\in(1,\phi)$.
%
\item If on the other hand $\alpha\geq\phi$, then $s\geq
2(1+1/\alpha)\geq 2(1+1/\alpha^2+1/\alpha^3)$, as
$1/\alpha \geq 1/\alpha^2 + 1/\alpha^3$ holds for $\alpha \geq \phi$.
\end{itemize}
Now we obtain
\begin{align*}
M-X &\geq \left(\frac{s}2-1\right)X-\left(1+\frac1\alpha\right)\ell_i \\
&\geq
\left(1+\frac1{\alpha^2}+\frac1{\alpha^3}-1\right)X-\left(1+\frac1\alpha\right)\ell_i \\
&\geq
\left(\frac1{\alpha^2}+\frac1{\alpha^3}\right)\alpha^2\ell_i-\left(1+\frac1\alpha\right)\ell_i=0 \enspace,
\end{align*}
and \eqref{eq:master1} holds.

\mycase{v}
$X\geq 2\ell_i$ and $\alpha<\alpha_1$.
(Note that this includes all cases when at least
two packets contribute to $X$, but we use it only if $\alpha<\alpha_1$.)
Using $s\geq 3+1/\alpha$
we obtain
\[
M-X\geq
\left(\frac12\left(3+\frac1\alpha\right)-1\right)X-\left(1+\frac1\alpha\right)\ell_i
\geq
\frac12\left(1+\frac1\alpha\right)2\ell_i-\left(1+\frac1\alpha\right)\ell_i
=0 \enspace,
\]
and \eqref{eq:master1} holds.

\smallskip
\noindent
{\bf Proof for $\alpha<\alpha_1$:}
We now observe that for $\alpha<\alpha_1$, we have exhausted all
the possible values of $X$. Indeed, if (v) does not apply, then at
most a single packet contributes to $X$, and one of the cases (i)-(iv)
applies, as (iv) covers the case when $X\geq\ell_{i+2}$,
and as $X=\ell_{i+1}$ is covered by (iii) or (v). Thus
\eqref{eq:master1} holds and the proof is complete.

\smallskip
\noindent
{\bf Proof for $\alpha\geq\alpha_1$:} We now analyze the remaining cases
for $\alpha\geq\alpha_1$.

\mycase{vi}
$X\geq (\alpha+1)\ell_i$.
(Note that this includes all cases when two packets not both of size
$\ell_i$ contribute to $X$.)
Using $s\geq 2+2/\alpha$ we obtain
$$M-X \geq\left(1+\frac{1}{\alpha}-1\right)X-\left(1+\frac1\alpha\right)\ell_i
\geq \frac{1}{\alpha} (\alpha+1)\ell_i -\left(1+\frac1\alpha\right)\ell_i
= 0$$
and \eqref{eq:master1} holds.

\mycase{vii}
$X=n\cdot\ell_i<(\alpha+1)\ell_i$ for some $n=2,3,\ldots$.
Since $\alpha>\alpha_1 > \phi$, we have $\ell_{i+1}>\ell_i+\ell_{i-1}$. This
implies that the first packet of size at least $\ell_i$ that is
scheduled in the phase has size equal to $\ell_i$ by the condition in
Step (3) of the algorithm. Thus, if also a packet of size larger than
$\ell_i$ contributes to $L_\MAINs(\geq i,(u,v])$, we have
$$L_\MAINs(\geq i,(u,v])\geq \ell_{i+1}+\ell_i\geq(\alpha+1)\ell_i>X$$
by the case condition and \eqref{eq:master1} holds. Otherwise
$L_\MAINs(\geq i,(u,v])$ is a multiple of $\ell_i$. Using $s\geq
2+2/\alpha$, we obtain
$$M\ge \left(1+\frac1\alpha\right)n\cdot\ell_i -\left(1+\frac1\alpha\right)\ell_i\geq(n-1)\left(1+\frac1\alpha\right)\ell_i>(n-1)\ell_i\,.$$
This, together with
divisibility by $\ell_i$ implies $L_\MAINs(\geq i,(u,v])\geq
n\cdot\ell_i=X$ and \eqref{eq:master1} holds again.

\mycase{viii}
$X=\ell_{i+1}$ and $\ell_{i+1}>2\ell_i$. We distinguish two
subcases depending on the size of the unfinished packet of \MAINs in
this phase.

If the unfinished packet has size at most $\ell_{i+1}$, the size of
the completed packets is bounded by
$$L_\MAINs((u,v])>sX-\ell_{i+1}=(s-1)\ell_{i+1}\geq \left(1+\frac{2}{\alpha}\right)\ell_{i+1}\,,$$
using $s\geq 2+2/\alpha$. Since the total size
of packets smaller than $\ell_i$ is less then $(1+1/\alpha)\ell_i$ by 
Lemma~\ref{l:main}(ii), we obtain
$$L_\MAINs(\geq i,(u,v])-X>\frac{2\ell_{i+1}}{\alpha}-\left(1+\frac1\alpha\right)\ell_i
\ge 2\ell_i -\left(1+\frac1\alpha\right)\ell_i >0\,,$$
where the penultimate inequality uses $\ell_{i+1}/\alpha\geq\ell_i$. Thus \eqref{eq:master1} holds.

Otherwise the unfinished packet has size at least $\ell_{i+2}$ and, by
Step (3) of the algorithm, also $L_\MAINs((u,v])>\ell_{i+2}$. We
have $\ell_{i+2}\geq\alpha\ell_{i+1}$  and by the case condition
$\ell_{i+1}>2\ell_i$ we obtain
$$L_\MAINs(\geq i,(u,v])-X>(\alpha-1)\ell_{i+1}-\left(1+\frac1\alpha\right)\ell_i>
2(\alpha-1)\ell_i-\left(1+\frac1\alpha\right)\ell_i\geq 0\,,$$
as the definition of $\alpha_1$ implies that $2(\alpha-1)\geq 1+1/\alpha$ for
$\alpha\geq\alpha_1$. Thus \eqref{eq:master1} holds.

\smallskip
We now observe that we have exhausted all the possible values of $X$
for $\alpha\geq\alpha_1$. Indeed, if at least two packets contribute to
$X$, either (vi) or (vii) applies.  Otherwise, at most a single
packet contributes to $X$, and one of the cases (i)-(iv) or (viii)
applies, as (iv) covers the case when $X\geq\ell_{i+2}$. Thus
\eqref{eq:master1} holds and the proof is complete.
\end{proof}

\subsubsection{Divisible Packet Sizes}\label{sec:main-divisible}

Now, we turn briefly to even more restricted \emph{divisible instances} considered by
Jurdziński et al.~\cite{JurdzinskiKL13}, which are a special case of $2$-separated instances.
Namely, we improve upon Theorem~\ref{thm:main-separated} in Theorem~\ref{thm:main-divisible}
presented below in the following sense: While the former guarantees that \MAINs is
$1$-competitive on (more general) $2$-separated instances at speed $s \geq 3$,
the latter shows that speed $s \geq 2.5$ is sufficient for 
(more restricted) divisible instances.
Moreover, we note that that by an example in Section~\ref{sec:examples},
the bound of Theorem~\ref{thm:main-divisible} is tight, i.e.,
\MAINs is not $1$-competitive for $s<2.5$, even on divisible instances.

\begin{theorem}
	\label{thm:main-divisible}
	If the packet sizes are divisible, then 
	$\MAINs$ is 1-competitive for $s\geq 2.5$.
\end{theorem}

\begin{proof}
	Lemma~\ref{l:main-segments}(i) implies \eqref{eq:master-init}.  We now
	prove \eqref{eq:master} for any proper $i$-segment $(u,v]$ with
	$v-u\geq\ell_i$ and $R=1$. The bound then follows by the
	Master Theorem. Since there is a fault at time $u$,
	we have $L_\OPT(\geq i,(u,v])\leq v-u$.
	
	By divisibility we have $L_\OPT(\geq i,(u,v])=n\ell_i$ for some
	nonnegative integer $n$. We distinguish two cases based on the size
	of the last packet started by \MAIN in the $i$-segment $(u,v]$, which
	is possibly unfinished due to a fault at $v$.
	
	If the unfinished packet has size at most $n\ell_i$, then
	$$L_\MAINs(\geq i,(u,v])>5(v-u)/2-\ell_i-\ell_{i-1}-n\ell_i
	\geq 5n\ell_i/2-3\ell_i/2-n\ell_i\geq (n-1)\ell_i$$
	by Lemma~\ref{lem:simple} and Lemma~\ref{l:main}(ii). Divisibility now implies
	$L_\MAINs(\geq i,(u,v])\geq n\ell_i=L_\OPT(\geq i,(u,v])$.
	
	Otherwise, by divisibility the size of the unfinished packet is at
	least $(n+1)\ell_i$ and the size of the completed packets is
	larger by the condition in Step (3) of the algorithm; here we also use the
	fact that $\MAINs$ completes the packet started at $u$, as its size is
	at most $\ell_i\leq v-u$ (otherwise, $u$ would be $i$-good, thus $C_i\ge u$
	and $(u,v]$ is not a proper $i$-segment). Thus $L_\MAINs(\geq i,(u,v])>(n+1)\ell_i-3\ell_i/2 \geq
	(n-1/2)\ell_i$.  Divisibility again implies $L_\MAINs(\geq i,(u,v])\geq
	n\ell_i=L_\OPT(\geq i,(u,v])$, which shows \eqref{eq:master}.
	%
	\end{proof}

\subsection{Some Examples for {\MAIN}} \label{sec:examples}


\subsubsection{General Packet Sizes}

\paragraph{Speeds below 2}

We show an instance on which the performance of \MAINs matches the bound of Theorem~\ref{thm:gen}.

\begin{remark}
	{\MAINs} has competitive ratio at least $1 + 2/s$ for $s < 2$.
\end{remark}

\begin{proof}
Choose a large enough integer $N$.
At time 0 the following packets are released:
$2N$ packets of size $1$, one packet of size $2$
and $N$ packet of size $4/s - \varepsilon$ for a small enough $\varepsilon > 0$
such that it holds $2 < 4/s - \varepsilon$.
These are all packets in the instance.

First there are $N$ phases, each of length $4/s - \varepsilon$ and ending by a fault.
{\OPT} completes a packet of size $4/s - \varepsilon$
in each phase, while {\MAINs} completes $2$ packets of size $1$
and then it starts a packet of size $2$ which is not finished.

Then there is a fault every $1$ unit of time, so that 
{\OPT} completes all packets of size 1, while the algorithm
has no pending packet of size 1 and as $s < 2$ the length of the phase is
not sufficient to finish a longer packet.

Overall, {\OPT} completes packets of total size $2 + 4/s - \varepsilon$
per phase, while the algorithm completes packets of total size only $2$ per phase.
The ratio thus tends to $1 + 2/s$ as $\varepsilon \rightarrow 0$.
 \end{proof}

\paragraph{Speeds between 2 and 4}

Now we show an instance which proves that \MAINs is not $1$-competitive for
$s<4$.  In particular, this implies that the speed sufficient for
$1$-competitiveness in Theorem~\ref{thm:four}, which we prove later,
cannot be improved.

\begin{remark}
	{\MAINs} has competitive ratio at least $4/s>1$ for $s \in [2,4)$.
\end{remark}

\begin{proof}
Choose a large enough integer $y$.
There will be four packet sizes: $1, x, y$ and $z$ such that $1 < x < y < z$,
$z = x+y-1$, and $x = y\cdot (s-2) / 2 + 2$; 
as $s\ge 2$ it holds $x > 1$ and
as $s < 4$ we have $x \leq y-1$ for a large enough $y$.

We will have $N$ phases again.
At time 0 the adversary releases all $N(y-1)$ packets of size $1$,
all $N$ packets of size $y$ and a single packet of size $z$
(never completed by either {\OPT} or {\MAINs}), whereas the packets of size $x$ are 
released one per phase. 

In each phase {\MAINs} completes, in this order:
$y-1$ packets of size $1$ and then a packet of size $x$, which has arrived just after
the $y-1$ packets of size $1$ are completed.
Next, it will start a packet of size $z$ and fail due to a jam.
We show that {\OPT} will complete a packet of size $y$. 
To this end, it is required that $y < 2(x+y-1) / s$, or equivalently
$x > y\cdot (s-2) / 2 + 1$ which holds by the choice of $x$.

After these $N$ phases, we will have jams every $1$ unit of time,
so that {\OPT} can complete all the $N(y-1)$ packets of size $1$,
while {\MAINs} will be unable to complete any packet (of size $y$ or larger).
The ratio per phase is
\[
\frac{\OPT}{\MAINs} = \frac{y-1 + y}{y-1 + x}
= \frac{2y-1}{y-1 + \frac{y\cdot (s-2)}{2} + 2}
= \frac{2y-1}{\frac{y\cdot s}{2} + 1}
\]
which tends to $4/s$ as $y \to \infty$.

\end{proof}

This example also disproves the claim
of Anta \etal~\cite{Anta13-dual} that their $(m,\beta)$-LAF algorithm is $1$-competitive
at speed $3.5$, even for one channel, i.e., $m=1$, where it behaves almost exactly as {\MAINs} ---
the sole difference is that LAF starts a phase by choosing a ``random'' packet.
As this algorithm is deterministic, we understand this to mean ``arbitrary'', so in
particular the same as chosen by {\MAINs}.

\subsubsection{Divisible Case}

We give an example that shows that \MAIN\ is not
very good for divisible instances, in particular it is not $1$-competitive
for any speed $s<2.5$ and thus the bound in Theorem~\ref{thm:main-divisible}
is tight. 

\begin{remark}
	{\MAINs} has competitive ratio at least $4/3$ on divisible instances if $s < 2.5$.
\end{remark}

\begin{proof}
We use packets of sizes $1$,
$\ell$, and $2\ell$ and we take $\ell$ sufficiently large compared to the
given speed or competitive ratio. There are many packets of size $1$
and $2\ell$ available at the beginning, the packets of size $\ell$
arrive at specific times where \MAIN\ schedules them immediately.

The faults occur at times divisible by $2\ell$, so
the optimum schedules one packet of size $2\ell$ in each phase between
two faults. We have $N$ such phases, $N(2\ell-1)$ packets of size $1$
and $N$ packets of size $2\ell$ available at the beginning. In each phase, $\MAINs$ schedules
$2\ell-1$ packets of size $1$, then a packet of size $\ell$ arrives
and is scheduled, and then a packet of size $2\ell$ is scheduled. The
algorithm would need speed $2.5-1/(2\ell)$ to complete it. So, for
$\ell$ large, the algorithm completes only packets of total size
$3\ell-1$ per phase. After these $N$ phases, we have faults every 1 unit of time, so the optimum
schedules all packets of size $1$, but the algorithm has no packet of size 1 pending
and it is unable to finish a longer packet.
The optimum thus finishes all packets $2\ell$ plus all small packets,
a total of $4\ell-1$ per phase. Thus the ratio tends to $4/3$ as $\ell\to\infty$.

\end{proof}

\subsection{Algorithm {\DIV} and its Analysis}
\label{sec:div}

We introduce our other algorithm \DIV designed for divisible instances.
Actually, it is rather a fine-tuned version of \MAIN, as it differs from it only in Step (3),
where \DIV enforces an additional {\it divisibility condition}, set apart by italics
in its formalization below.
Then, using our framework of local analysis from this section,
we give a simple proof that \DIV matches the performance of the algorithms
from~\cite{JurdzinskiKL13} on divisible instances.

\begin{center}
	\fbox{\parbox{0.9\textwidth}{
			{\bf Algorithm \DIV}
			\begin{enumerate}[label=(\arabic*), nosep]
				\item 
				If no packet is pending, stay idle until the next release time.
				\item
				Let $i$ be the maximal $i \leq k$ such that there is a pending packet
				of size $\ell_i$ and $\ell(P^{<i})<\ell_i$.  Schedule a packet of size
				$\ell_i$ and set $t_B = t$.
				\item
				Choose the maximum $i$ such that
				\\
				\mbox{}
				\quad
				(i) there is a pending packet of size $\ell_i$, 
				\\
				\mbox{}
				\quad
				(ii) $\ell_i \leq \rel(t)$ and
				\\
				\mbox{}
				\quad
				{\it (iii) $\ell_i$ divides $\rel(t)$.}
				\\
				Schedule a packet of size $\ell_i$.  Repeat Step (3) as long as such $i$
				exists.
				\item
				If no packet satisfies the condition in Step (3), go to Step (1). 
			\end{enumerate}
	}}
\end{center}

Throughout the section we assume that the packet sizes are divisible.
We note that Lemmata~\ref{lem:simple} and~\ref{l:master-small} and the Master Theorem 
apply to \DIV as well, since their proofs are not influenced by the divisibility condition.
In particular, the definition of critical times $C_i$
(Definition~\ref{def:critical}) remains the same.
Thus, this section is devoted to leveraging divisibility
to prove stronger stronger analogues of Lemma~\ref{l:main} and Lemma~\ref{l:main-segments}
(which are not needed to prove the Master Theorem) in this order.
Once established, these are combined with the Master Theorem to prove
that $\DIV(2)$ is $1$-competitive and $\DIV(1)$ is $2$-competitive.
Recall that $\rel(t) = s\cdot (t - t_B)$ is the relative time after the start
of the current phase $t_B$, scaled by the speed of the algorithm.

\begin{lemma}
	\label{l:div-obs}
	\begin{enumerate}[label=\rm(\roman*), nosep]
		\item
		If \DIV\ starts or completes a packet of size $\ell_i$ at time $t$,
		then $\ell_i$ divides $\rel(t)$.
		\item
		Let $t$ be a time with $\rel(t)$ divisible by $\ell_i$ and
		$\rel(t)>0$. If a packet of size $\ell_i$ is pending at time $t$, then
		\DIV\ starts or continues running a packet of size at least $\ell_i$
		at time $t$.
		\item
		If at the beginning of phase at time $u$ a packet of size $\ell_i$
		is pending and no fault occurs before time $t=u+\ell_i/s$, then the
		phase does not end before $t$. 
	\end{enumerate}
\end{lemma}
\begin{proof}
	(i) follows trivially from the description of the algorithm.
	
	(ii): If \DIV\ continues running  some packet at $t$, it cannot be a packet
	smaller than $\ell_i$ by (i) and the claim follows. If \DIV\
	starts a new packet, then a packet of size $\ell_i$ is pending by the
	assumption. Furthermore, it satisfies all the conditions from Step 3
	of the algorithm, as $\rel(t)$ is divisible by $\ell_i$ and
	$\rel(t)\geq\ell_i$ (from $\rel(t)>0$ and divisibility).
	Thus the algorithm starts a packet of size at least $\ell_i$.
	
	(iii): We proceed by induction on $i$. Assume that no fault happens
	before $t$. If the phase starts by a packet of size at least $\ell_i$,
	the claim holds trivially, as the packet is not completed before
	$t$. This also proves the base of the induction for $i=1$.
	
	It remains to handle the case when the phase starts by a packet
	smaller than $\ell_i$. Let $P^{<i}$ be the set of all packets of size
	smaller than $\ell_i$ pending at time $u$. By the Step (2) of the
	algorithm, $\ell(P^{<i})\geq\ell_i$. We
	show that all packets of $P^{<i}$ are completed if no fault happens, which
	implies that the phase does not end before $t$.
	
	Let $j$ be such that $\ell_j$ is the maximum size of a packet in $P^{<i}$;
	note that $j$ exists, as the phase starts by a packet smaller than
	$\ell_i$.  By the induction assumption, the phase does not end before
	time $t'=u+\ell_j/s$.  From time $t'$ on, the conditions in Step (3)
	guarantee that the remaining packets from $P^{<i}$ are processed
	from the largest ones, possibly interleaved with some of the newly
	arriving packets of larger sizes, as $\rel(\tau)$ for the current
	time $\tau\ge t'$ such that a packet completes at $\tau$ is always divisible by the size of the largest pending
	packet from $P^{<i}$. This shows that the phase cannot end before all
	packets from $P^{<i}$ are completed if no fault happens.
	%
\end{proof}

Now we prove a stronger analogue of Lemma~\ref{l:main-segments}.
\begin{lemma}
	\label{l:div-critical}
	\begin{enumerate}[label=\rm(\roman*), nosep]
		\item
		If $(u,v]$ is the initial $i$-segment,
		then $$L_\DIVs(\geq i,(u,v])>s(v-u)-3\ell_k\,.$$
		\item
		If $(u,v]$ is a proper $i$-segment and $v-u\geq \ell_i$ then $$L_\DIVs(\geq
		i,(u,v])>s(v-u)/2-\ell_i\,.$$ Furthermore, $L_\DIVs((u,v])>s(v-u)/2$ and
		$L_\DIVs((u,v])$ is divisible by $\ell_i$.
	\end{enumerate}
\end{lemma}
\begin{proof}
	Suppose that time $t\in[u,v)$ satisfies that $\rel(t)$ is divisible by
	$\ell_i$ and $\rel(t)>0$. Then observe that Lemma~\ref{l:div-obs}(ii)
	together with the assumption that a packet of size $\ell_i$ is
	always pending in $[u,v)$ implies that from time $t$ on only packets
	of size at least $\ell_i$ are scheduled, and thus the current phase does not
	end before $v$. 
	
	For a proper $i$-segment $(u,v]$, the previous observation for
	$t=u+\ell_i/s$ immediately implies (ii): Observe that $t\leq v$ by the
	assumption of (ii). Now $L_\DIVs(<i,(u,v])$ is either equal to $0$ (if
	the phase starts by a packet of size $\ell_i$ at time $u$), or equal to
	$\ell_i$ (if the phase starts by a smaller packet). In both cases
	$\ell_i$ divides $L_\DIVs(<i,(u,v])$ and thus also $L_\DIVs((u,v])$. As
	in the analysis of \MAIN, the total size of completed packets is more
	than $s(v-u)/2$ and (ii) follows.
	
	For the initial $i$-segment $(u,v]$ we first observe that the claim is
	trivial if $s(v-u)\leq 2\ell_i$. So we may assume that
	$u+2\ell_i/s\leq v$. Now we distinguish two cases:
	\begin{enumerate}
		\item The phase of $u$ ends at some time $u'\leq u+\ell_i/s$: Then, by 
		Lemma~\ref{l:div-obs}(iii) and the initial observation, the phase that 
		immediately follows the one of $u$ does not end in $(u',v)$ and from time 
		$u'+\ell_i/s$ on, only packets of size at least $\ell_i$ are scheduled.
		Thus $L_\DIVs(<i,(u,v])\leq 2\ell_i$.
		\item The phase of $u$ does not end by time $u+\ell_i/s$:
		Thus there exists $t\in (u,u+\ell_i/s]$ such that $\ell_i$ divides
		$\rel(t)$ and also $\rel(t)>0$ as $t>u$. Using the initial
		observation for this $t$ we obtain that
		the phase does not end in $(u,v)$ and from time
		$t$ on only packets of size at least $\ell_i$ are
		scheduled. Thus $L_\DIVs(<i,(u,v])\leq\ell_i$.
	\end{enumerate}
	In both cases $L_\DIVs(<i,(u,v])\leq2\ell_i$, furthermore only a single
	packet is possibly unfinished at time $v$. Thus $L_\DIVs(\geq
	i,(u,v])>s(v-u)-2\ell_i-\ell_k$ and (i) follows.
\end{proof}

\begin{theorem}\label{thm:DIVresults}
	Let the packet sizes be divisible. Then $\DIV(1)$ is
	$2$-competitive. Also, for any speed $s\geq 2$, $\DIVs$ is
	$1$-competitive.
\end{theorem}
\begin{proof}
	Lemma~\ref{l:div-critical}(i) implies \eqref{eq:master-init}.  We now
	prove (\ref{eq:master}) for any proper $i$-segment $(u,v]$ with
	$v-u\geq\ell_i$ and appropriate $R$. The theorem then follows by the
	Master Theorem.
	
	Since $u$ is a time of a fault, we have $L_\OPT(\geq i,(u,v])\leq v-u$.
	If $L_\OPT(\geq i,(u,v])=0$, (\ref{eq:master}) is trivial. Otherwise
	$L_\OPT(\geq i,(u,v])\geq\ell_i$, thus $v-u\geq \ell_i$ and the
	assumption of Lemma~\ref{l:div-critical}(ii) holds.
	
	\bigskip
	For $s\geq 2$, Lemma~\ref{l:div-critical}(ii) implies
	\[
	L_\DIVs(\geq i,(u,v])>s(v-u)/2-\ell_i
	\geq v-u-\ell_i\geq L_\OPT(\geq i,(u,v])-\ell_i
	\,.
	\]
	Since both $L_\DIVs(\geq i,(u,v])$ and $L_\OPT(\geq i,(u,v])$ are
	divisible by $\ell_i$, this implies $L_\DIVs(\geq i,(u,v])\geq
	L_\OPT(\geq i,(u,v])$, i.e., (\ref{eq:master}) holds for $R=1$.
	
	\bigskip
	For $s=1$, Lemma~\ref{l:div-critical}(ii) implies
	\begin{align*}
		L_\DIV((u,v])+ L_\DIV(\geq i,(u,v])&>(v-u)/2+(v-u)/2-\ell_i\\
		&\geq v-u-\ell_i\geq L_\OPT(\geq i,(u,v])-\ell_i\,.
	\end{align*}
	Since $L_\DIV((u,v])$, $L_\DIV(\geq i,(u,v])$, and $L_\OPT(\geq
	i,(u,v])$ are all divisible by
	$\ell_i$, this implies
	$L_\DIV((u,v])+
	L_\DIV(\geq i,(u,v])\geq L_\OPT(\geq i,(u,v])$,
	i.e., (\ref{eq:master}) holds for $R=2$.
\end{proof}

\subsubsection{Example with Two Divisible Packet Sizes}\label{sec:example-2sizes}

We show that for our algorithms speed $2$ is necessary
if we want a ratio below $2$, even if there are only two packet sizes
in the instance. This matches the upper bound given in Theorem~\ref{thm:gen} for $\MAIN(2)$
and our upper bounds for $\DIV(s)$ on divisible instances, i.e., ratio $2$ for $s<2$
and ratio $1$ for $s \geq 2$. 
We remark that by Theorem~\ref{thm:LB2}, no deterministic algorithm can be 
$1$-competitive with speed $s<2$ on divisible instances, but this example shows
a stronger lower bound for our algorithms, namely that their ratios are at least $2$.

\begin{remark}
\MAIN and \DIV have ratio no smaller than $2$ when $s<2$, even if packet sizes are only $1$ and
$\ell \geq \max\{s+\epsilon,\ \epsilon / (2-s) \}$
for an arbitrarily small $\epsilon>0$.
\end{remark}

\begin{proof}
We denote either algorithm by {\ALG}.
There will be $N$ phases, that all look the same: In each phase, issue
one packet of size $\ell$ and $\ell$ packets of size 1, and have the phase
end by a fault at time $(2\ell-\varepsilon)/s \geq \ell$ which holds by 
the bounds on $\ell$.
Then {\ALG} will complete all $\ell$ packets of size $1$
but will not complete the one of size $\ell$.  By the previous inequality,
{\OPT} can complete the packet of size $\ell$ within the phase. 
Once all $N$ phases are over, the jams occur every $1$ unit of time,
which allows {\OPT} completing all $N\ell$ remaining packets of size $1$.
However, {\ALG} is unable to complete any of the packets of size $\ell$.
Thus the ratio is $2$.
 \end{proof}

\section{PrudentGreedy with Speed 4}
\label{sec:four}

In this section we prove that speed 4 is sufficient for \MAIN\ to be
1-competitive.  An example in Section~\ref{sec:examples} show that speed
4 is also necessary for our algorithm.

\begin{theorem}\label{thm:four}
$\MAINs$ is 1-competitive for $s\geq 4$.
\end{theorem}

\paragraph{Intuition}
For $s\ge 4$ we have that if at the start of a phase $\MAINs$ has a
packet of size $\ell_i$ pending and the phase has length at least
$\ell_i$, then $\MAINs$ completes a packet of size at least
$\ell_i$. To show this, assume that the phase starts at time $t$. Then
the first packet $p$ of size at least $\ell_i$ is started before time
$t+2\ell_i/s$ by Lemma~\ref{l:main}(ii) and by the condition in Step (3) it has size smaller
than $2\ell_i$. Thus it completes before time $t+4\ell_i/s\leq
t+\ell_i$, which is before the end of the phase. This property does
not hold for $s<4$. It is important in our proof, as it shows that if
the optimal schedule completes a job of some size, and such job is
pending for $\MAINs$, then $\MAINs$ completes a job of the same size
or larger. However, this is not sufficient to complete the proof by a
local (phase-by-phase) analysis similar to the previous section, as
the next example shows.

Assume that at the beginning, we release $N$ packets of size $1$, $N$
packets of size $1.5 - 2\eps$, one packet of size $3-2\eps$ and a sufficient
number of packets of size $1-\eps$, for a small $\eps>0$.  Our focus is on
packets of size at least $1$.  Supposing $s=4$ we have the following phases:
\begin{itemize}[nosep]
\item First, there are $N$ phases of length $1$. In each phase the optimum
  completes a packet of size $1$, while among packets of size at least $1$, 
  $\MAINs$ completes a packet of size $1.5 - 2\eps$, as it starts packets of sizes $1-\eps$, $1-\eps$, $1.5 - 2\eps$,
  $3-2\eps$, in this order, and the last packet is jammed.
\item Then there are $N$ phases of length $1.5 - 2\eps$ where the optimum completes a packet
  of size $1.5 - 2\eps$ while among packets of size at least $1$, the algorithm 
  completes only a single packet of size $1$, as it starts packets of sizes 
  $1-\eps$, $1-\eps$, $1$, $3-2\eps$, in this order.
  The last packet is jammed, since for $s=4$
  the phase must have length at least $1.5 - \eps$ to complete it.
\end{itemize}
In phases of the second type, the algorithm does not complete more (in terms of
total size) packets of size at least $1$ than the optimum.  Nevertheless, in our
example, packets of size $1.5 - 2\eps$ were already finished by the algorithm,
and this is a general rule.  The novelty in our proof is a complex charging 
argument that exploits such subtle interaction between phases.

\paragraph{Outline of the proof}
We define critical times $C'_i$ similarly as before, but without the
condition that they should be ordered
(thus either $C'_i \le C'_{i-1}$ or $C'_i > C'_{i-1}$ may hold).
Then, since the algorithm has nearly no pending packets of size $\ell_i$
just before $C'_i$, we can charge almost all adversary's packets of size $\ell_i$ started
before $C'_i$ to algorithm's packets of size $\ell_i$ completed
before $C'_i$ in a 1-to-1 fashion;  we thus call these charges 1-to-1
charges. We account for the first few packets of each size
completed at the beginning of \ADV, the schedule of the adversary,
in the additive constant of the competitive ratio,
thereby shifting the targets of the 1-to-1 charges backward
in time. This also resolves what to do with the yet uncharged
packets pending for the algorithm just before $C'_i$.
%


After the critical time $C'_i$, packets of size $\ell_i$ are always
pending for the algorithm, and thus (as we observed above) the
algorithm schedules a packet of size at least $\ell_i$ when
the adversary completes a packet of size $\ell_i$.
It is actually more convenient not to work with phases,
but partition the schedule into blocks inbetween
successive faults. A block can contain several phases of the algorithm
separated by an execution of Step (4); however, in the most important
and tight part of the analysis the blocks coincide with phases.

In the crucial lemma of the proof, based on these observations and their refinements,
we show that we can assign the remaining packets in \ADV to algorithm's packets in the same block so
that for each algorithm's packet $q$ the total size of packets
assigned to it is at most $\ell(q)$. However, we cannot use this
assignment directly to charge the remaining packets, as some of the
algorithm's big packets may receive 1-to-1 charges, and in this case
the analysis needs to handle the interaction of different blocks. This
very issue can be seen even in our introductory example.

To deal with this, we process blocks in the order of time from the
beginning to the end of the schedule, simultaneously completing the
charging to the packets in the current block of the schedule of \MAINs and
possibly modifying \ADV in the future
blocks.  In fact, in the assignment described above, we include not
only the packets in \ADV without 1-to-1 charges, but also packets in
\ADV with a 1-to-1 charge to a later block. After creating the
assignment, if we have a packet $q$ in \MAIN that receives a 1-to-1
charge from a packet $p$ in a later block of \ADV, we remove $p$ from
\ADV in that later block and replace it there by the packets assigned
to $q$ (that are guaranteed to be of smaller total size than $p$).
After these swaps, the 1-to-1 charges together with the assignment
form a valid charging that charges the remaining not swapped packets
in \ADV in this block together with the removed packets from the later
blocks in \ADV to the packets of \MAINs in the current block. This
charging is now independent of the other blocks, so we can continue
with the next block.


\subsection{Blocks, Critical Times, 1-to-1 Charges and the Additive Constant} 
We now formally define the notions of blocks and (modified) critical times.

\begin{definition}
Let $f_1, f_2,\ldots, f_N$ be the times of faults. Let $f_0=0$ and
$f_{N+1}=T$ is the end of schedule. Then the time interval
$(f_i, f_{i+1}]$, $i=0,\ldots,N$, is called a block. 
\end{definition}

\begin{definition}
For $i=1,\ldots k$, the critical time $C'_i$ is the supremum of
$i$-good times $t\in[0,T]$, where $T$ is the end of the schedule and
$i$-good times are as defined in Definition~\ref{def:critical}.
\end{definition}

All $C'_i$'s are defined, as $t=0$ is $i$-good for all $i$. Similarly
to Section~\ref{sec:master}, each $C'_i$ is of one of the following
types: (i) $C'_i$ starts a phase and a packet larger than $\ell_i$ is
scheduled, (ii) $C'_i=0$, (iii) $C'_i=T$, or (iv) just before time
$C'_i$ no packet of size $\ell_i$ is pending but at time $C'_i$ one or
more packets of size $\ell_i$ are pending; in this case $C'_i$ is not
$i$-good but only the supremum of $i$-good times. We observe that in
each case, at time $C'_i$ the total size of packets $p$ of size
$\ell_i$ pending for $\MAINs$ and released before $C'_i$ is less than
$\ell_k$.

Next we define the set of packets that contribute to the additive constant.
\begin{definition}
Let the set $A$ contain for each $i = 1, \dots, k$:
\begin{enumerate}[label=\rm(\roman*), nosep]
\item the first $\lceil 4\ell_k / \ell_i \rceil$ packets of size $\ell_i$
  completed by the adversary, and
\item the first $\lceil 4\ell_k / \ell_i \rceil$ packets of size $\ell_i$
completed by the adversary after $C'_i$.
\end{enumerate}
If there are not sufficiently many packets of size $\ell_i$ completed
by the adversary in (i) or (ii), we take all the packets in (i) or all the packets completed
after $C'_i$ in (ii), respectively.
\end{definition}
For each $i$, we put into $A$  packets of size $\ell_i$ of total size
at most $10\ell_k$. Thus we have $\ell(A) = \Oh(k \ell_k)$ 
which implies that packets in $A$ can be counted in the additive
constant.

We define 1-to-1 charges for packets of size $\ell_i$ as follows. Let
$p_1$, $p_2$, \ldots, $p_n$ be all the packets of size $\ell_i$
started by the adversary before $C'_i$ that are not in $A$. We claim
that $\MAINs$ completes at least $n$ packets of size $\ell_i$ before
$C'_i$ if $n\ge 1$. Indeed, if $n\ge 1$, before time $C'_i$ at least $n+\lceil 4\ell_k / \ell_i
\rceil$ packets of size $\ell_i$ are started by the adversary and thus
released; by the definition of $C'_i$ at time $C'_i$ fewer than
$\ell_k/\ell_i$ of them are pending for $\MAINs$, one may be running and
the remaining ones must be completed. We now charge each $p_m$ to the
$m$th packet of size $\ell_i$ completed by $\MAINs$.
Note that each packet started by the adversary is charged at most once
and each packet completed by $\MAINs$ receives at most one charge.

\begin{figure}
  \begin{center}
    \includegraphics[width=0.65\textwidth]{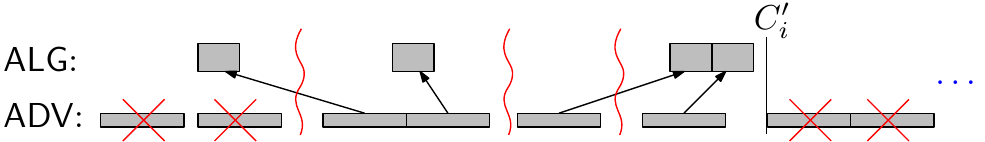}
  \end{center}
  \caption{An illustration of back, up, and forward 1-to-1 charges for $\ell_i$-sized packets
  (other packets are not shown). 
  The winding lines depict the times of jamming errors, i.e., the beginnings and ends of blocks.
  Note that the packets in the algorithm's schedule are shorter, but wider, which illustrates
  that the algorithm runs the packets with a higher speed for a shorter time (the area
  thus corresponds to the amount of work done).
  Crossed packets are included in the set $A$ (and thus contribute to the additive constant).}
  \label{fig:1to1charges}
\end{figure}

We call a 1-to-1 charge starting and ending in the same block an
\textit{up charge}, a 1-to-1 charge from a block starting at $u$ to a
block ending at $v' \le u$ a \textit{back charge}, and a 1-to-1 charge
from a block ending at $v$ to a block starting at $u' \ge v$ a
\textit{forward charge}; see Figure~\ref{fig:1to1charges} for an illustration.
A \textit{charged packet} is a packet charged
by a 1-to-1 charge. The definition of $A$ implies the following
two important properties.

\begin{lemma}\label{l:forward}
Let $p$ be a packet of size $\ell_i$, started by the adversary at time $t$,
charged by a forward charge to a packet $q$ started by
$\MAINs$ at time $t'$. Then at
any time $\tau\in[t-3\ell_k,t')$, more than $\ell_k/\ell_i$ packets of
size $\ell_i$ are pending for $\MAINs$.
\end{lemma}
\begin{proof}
Let $m$ be the number of packets of size $\ell_i$ that $\MAINs$ completes
before $q$. Then, by the definition of $A$, the adversary completes
$m+\lceil 4\ell_k / \ell_i \rceil$ packets of size $\ell_i$ before
$p$. As fewer than $3\ell_k/\ell_i$ of these packets are started in
$(t-3\ell_k,t]$, the remaining more than $m+\ell_k / \ell_i$ packets
have been released before or at time $t-3\ell_k$. As only $m$ of them
are completed by $\MAINs$ before $t'$, the remaining more than $\ell_k /
\ell_i$ packets are pending at any time $\tau\in[t-3\ell_k,t')$.
\end{proof}

\begin{lemma}\label{l:notcharged}
Let $p\not\in A$ be a packet of size $\ell_i$ started by the
adversary at time $t$ that is not charged. Then $t-4\ell_k\geq C'_i$
and thus at any $\tau\geq t-4\ell_k$, a packet of size $\ell_i$ is
pending for $\MAINs$.
\end{lemma}
\begin{proof}
Any packet of size $\ell_i$ started before $C'_i+4\ell_k$ is either
charged or put in $A$, thus $t-4\ell_k\geq C'_i$. After $C'_i$, a
packet of size $\ell_i$ is pending by the definition of $C'_i$.
\end{proof}

\subsection{Processing Blocks}
Initially, let \ADV\ be an optimal (adversary) schedule. First, we
remove all packets in $A$ from \ADV.  Then we process blocks one by
one in the order of time.  When we process a block, we modify \ADV\ so
that we (i) remove some packets from \ADV, so that the total size of
removed packets is at most the total size of packets completed by
$\MAINs$ in this block, and (ii) reschedule any remaining packet in
\ADV\ in this block to one of the later blocks, so that the schedule
of remaining packets is still feasible. Summing over all blocks, (i)
guarantees that $\MAINs$ is $1$-competitive with an additive constant
$\ell(A)$.

When we reschedule a packet in \ADV, we keep the packet's 1-to-1
charge (if it has one), however, its type may change due to
rescheduling. Since we are moving packets to later times only,
the release times are automatically respected.  Also it follows that
we can apply Lemmata~\ref{l:forward} and~\ref{l:notcharged} even to
\ADV after rescheduling.

After processing of a block, there will remain no charges to or from
it. For the charges from the block, this is automatic, as \ADV contains
no packet in the block after we process it. For the charges to the block, this is
guaranteed as in the process we remove from \ADV all the packets in
later blocks charged by back charges to the current block.

From now on, let $(u,v]$ be the current block that we are processing;
all previous blocks ending at $v' \le u$ are processed. As there are
no charges to the previous blocks, any packet scheduled in \ADV\ in
$(u,v]$ is charged by an up charge or a forward charge, or else it is
not charged at all.  We distinguish two main cases of the proof,
depending on whether $\MAINs$ finishes any packet in the current block.

\subsubsection{Main Case 1: Empty Block}
The algorithm does not finish any packet in $(u,v]$. We claim that
\ADV\ does not finish any packet. The processing of the block is
then trivial.

For a contradiction, assume that \ADV\ starts a packet $p$ of size
$\ell_i$ at time $t$ and completes it. The packet $p$ cannot be
charged by an up charge, as $\MAINs$ completes no packet in this
block. Thus $p$ is either charged by a forward charge or not
charged. Lemma~\ref{l:forward} or~\ref{l:notcharged} implies
that at time $t$ some packet of size $\ell_i$ is pending for $\MAINs$.

Since \MAIN\ does not idle unnecessarily, this means that some packet
$q$ of size $\ell_j$ for some $j$ is started in $\MAINs$ at time
$\tau\leq t$ and running at $t$. As $\MAINs$ does not complete any
packet in $(u,v]$, the packet $q$ is jammed by the fault at time
$v$. This implies that $j>i$, as $\ell_j>s(v-\tau) \geq v-t\geq
\ell_i$; we also have $t-\tau<\ell_j$.  Moreover, $q$ is the
only packet started by $\MAINs$ in this block, thus it starts a phase.

As this phase is started by packet $q$ of size $\ell_j>\ell_i$, the
time $\tau$ is $i$-good and $C'_i\geq\tau$. All packets \ADV\ started
before time $C'_i+4\ell_k/s$ are charged, as the packets in $A$ are
removed from \ADV\ and packets in \ADV are rescheduled only to later
times. Packet $p$ is started before $v<\tau+\ell_j/s<C'_i+\ell_k/s$,
thus it is charged. It follows that $p$ is charged by a forward
charge. We now apply Lemma~\ref{l:forward} again and observe that it
implies that at $\tau>t-\ell_j$ there are more than $\ell_k/\ell_i$
packets of size $\ell_i$ pending for $\MAINs$. This is in contradiction
with the fact that at $\tau$, $\MAINs$ started a phase by $q$ of size
$\ell_j>\ell_i$.

\subsubsection{Main Case 2: Non-empty Block}
Otherwise, $\MAINs$ completes a packet in the current block $(u, v]$.

Let $Q$ be the set of packets completed by $\MAINs$ in $(u,v]$ that do
not receive an up charge. Note that no packet in $Q$ receives a
forward charge, as the modified \ADV\ contains no packets before $u$,
so packets in $Q$ either get a back charge or no charge at all.  Let
$P$ be the set of packets completed in \ADV\ in $(u,v]$ that are not
charged by an up charge. Note that $P$ includes packets charged by a
forward charge and uncharged packets, as no packets are charged to a
previous block.

We first assign packets in $P$ to packets in $Q$ so that for each
packet $q\in Q$ the total size of packets assigned to $q$ is at most
$\ell(q)$.  Formally, we iteratively define a provisional assignment
$f: P \rightarrow Q$ such that $\ell(f^{-1}(q)) \le \ell(q)$ for each
$q\in Q$.

\paragraph{Provisional assignment}
We maintain a set $O\subseteq Q$ of \textit{occupied} packets
that we do not use for a future assignment.  Whenever we assign a
packet $p$ to $q\in Q$ and $\ell(q) - \ell(f^{-1}(q)) <
\ell(p)$, we add $q$ to $O$.  This rule guarantees that each packet
$q\in O$ has $\ell(f^{-1}(q))>\ell(q)/2$.

We process packets in $P$ in the order of decreasing sizes as
follows. We take the largest unassigned packet $p\in P$ of size $\ell(p)$
(if there are more unassigned packets of size $\ell(p)$, we take an arbitrary one)
and choose an arbitrary packet $q\in Q\setminus O$ such that
$\ell(q)\ge \ell(p)$; we prove in Lemma~\ref{l:assignment} below that
such a $q$ exists. We assign $p$ to $q$, that is, we set $f(p) =
q$. Furthermore, as described above, if $\ell(q) - \ell(f^{-1}(q))<
\ell(p)$, we add $q$ to $O$. We continue until all packets are assigned.

If a packet $p$ is assigned to $q$ and $q$ is not put in $O$, it
follows that $\ell(q)-\ell(f^{-1}(q))\geq\ell(p)$. This implies that
after the next packet $p'$ is assigned to $q$, we have
$\ell(q)\geq\ell(f^{-1}(q))$, as the packets are processed from the
largest one and thus $\ell(p')\leq\ell(p)$. If follows that at the end
we obtain a valid provisional assignment.

\begin{lemma}
\label{l:assignment}
The assignment process above assigns all packets in $P$.
\end{lemma}

\begin{proof}
For each size $\ell_j$ we show that all packets of size $\ell_j$ in
$P$ are assigned, which is clearly sufficient. We fix the size
$\ell_j$ and define a few quantities.

Let $n$ denote the number of packets of size $\ell_j$ in $P$.  Let
$o$ denote the total \textit{occupied size}, defined as $o=\ell(O) +
\sum_{q\in Q\setminus O} \ell(f^{-1}(q))$ at the time just before we
start assigning the packets of size $\ell_j$. Note that the rule for
adding packets to $O$ implies that $\ell(f^{-1}(Q))\geq o/2$. Let $a$
denote the current total \textit{available size} defined as $a =
\sum_{q\in Q\setminus O: \ell(q)\ge \ell_j} (\ell(q)-\ell(f^{-1}(q)))$.
We remark that in the definition of $a$ we restrict attention only to packets
of size $\ge \ell_j$, but in the definition of $o$ we consider all packets in $Q$;
however, as we process in the order of decreasing sizes,
so far we have assigned packets from $P$ only to packets of size $\ge \ell_j$ in $Q$.

First, we claim that it is sufficient to show that $a>(2n-2)\ell_j$
before we start assigning the packets of size $\ell_j$.  As long as
$a>0$, there is a packet $q\in Q\setminus O$ of size at least $\ell_j$
and thus we may assign the next packet (and, as noted before, actually
$a\geq \ell_j$, as otherwise $q\in O$). Furthermore, assigning a
packet $p$ of size $\ell_j$ to $q$ decreases $a$ by $\ell_j$ if $q$ is
not added to $O$ and by less than $2\ell_j$ if $q$ is added to
$O$. Altogether, after assigning the first $n-1$ packets, $a$
decreases by less than $(2n-2)\ell_j$, thus we still have $a>0$, and
we can assign the last packet. The claim follows.

We now split the analysis into two cases, depending on whether there
is a packet of size $\ell_j$ pending for $\MAINs$ at all times in $[u,v)$,
or not.  In either case, we prove that the available space $a$ is
sufficiently large before assigning the packets of size $\ell_j$.

In the first case, we suppose that a packet of size $\ell_j$ is
pending for $\MAINs$ at all times in $[u,v)$.  Let $z$ be the total size
of packets of size at least $\ell_j$ charged by up charges in this
block.  The size of packets in $P$ already assigned is at least
$\ell(f^{-1}(Q))\geq o/2$ and we have $n$ yet unassigned packets of
size $\ell_j$ in $P$. As \ADV\ has to schedule all these packets and
the packets with up charges in this block, its size satisfies $v-u\ge
\ell(P)+z\ge n\ell_j+o/2+z$. Now consider the schedule of $\MAINs$ in
this block.  By Lemma~\ref{l:main}, there is no end of phase in
$(u,v)$ and jobs smaller than $\ell_j$ scheduled by $\MAINs$ have
total size less than $2\ell_j$. All the other completed packets
contribute to one of $a$, $o$, or $z$.  Using Lemma~\ref{lem:simple},
the previous bound on $v-u$ and $s\geq 4$, the total size of completed
packets is at least $s(v-u)/2 \ge 2n\ell_j+o+2z$.  Hence
$a>(2n\ell_j+o+2z)-2\ell_j-o-z\geq (2n-2)\ell_j$, which completes the
proof of the lemma in this case.

Otherwise, in the second case, there is a time in $[u,v)$ when no
packet of size $\ell_j$ is pending for $\MAINs$.
Let $\tau$ be the supremum of times $\tau'\in[u,v]$ such that
$\MAINs$ has no pending packet of size at least $\ell_j$ at time
$\tau'$; if no such $\tau'$ exists we set $\tau=u$.
Let $t$ be the time when the adversary starts the first packet $p$ of
size $\ell_j$ from $P$.

Since $p$ is charged using a forward charge or $p$ is not charged, we
can apply Lemma~\ref{l:forward} or~\ref{l:notcharged}, which implies
that packets of size $\ell_j$ are pending for $\MAINs$ from time
$t-3\ell_k$ till at least $v$. By the case condition, there is a time
in $[u,v)$ when no packet of size $\ell_j$ is pending, and this time
is thus before $t-3\ell_k$, implying $u<t-3\ell_k$. The definition
of $\tau$ now implies that $\tau \le t-3\ell_k$.

Towards bounding $a$, we show that (i) $\MAINs$ runs a limited amount of
small packets after $\tau$ and thus $a+o$ is large, and that (ii)
$f^{-1}(Q)$ contains only packets run by \ADV\ from $\tau$ on, and
thus $o$ is small.

We claim that the total size of packets smaller than $\ell_j$
completed in $\MAINs$ in $(\tau,v]$ is less than $3\ell_k$. This claim is
similar to Lemma~\ref{l:main} and we also argue similarly.  Let
$\tau_1<\tau_2<\ldots<\tau_\alpha$ be all the ends of phases in
$(\tau,v)$ (possibly there is none, then $\alpha=0$); also let
$\tau_0=\tau$. For $i=1,\ldots,\alpha$, let $r_i$ denote the packet
started by $\MAINs$ at $\tau_i$; note that $r_i$ exists since after
$\tau$ there is a pending packet at any time in $[\tau,v]$ by the
definition of $\tau$. See Figure~\ref{fig:s4taus} for an illustration.
First note that any packet started at
or after time $\tau_\alpha+\ell_k/s$ has size at least $\ell_j$, as
such a packet is pending and satisfies the condition in Step (3) of
the algorithm. Thus the total amount of the small packets completed in
$(\tau_\alpha,v]$ is less than $\ell_k+\ell_{k-1}<2\ell_k$. The claim
now follows for $\alpha=0$. Otherwise, as there is no fault in
$(u,v)$, at $\tau_i$, $i=1,\ldots,\alpha$, Step (4) of the algorithm
is reached and thus no packet of size at most $s(\tau_i-\tau_{i-1})$
is pending. In particular, this implies that
$\ell(r_i)>s(\tau_i-\tau_{i-1})$ for $i=1,\ldots,\alpha$. This also
implies that the amount of the small packets completed in
$(\tau_0,\tau_1]$ is less than $\ell_k$ and the claim for $\alpha=1$
follows. For $\alpha\geq 2$ first note that by Lemma~\ref{l:main}(i),
$s(\tau_i-\tau_{i-1})\geq\ell_j$ for all $i=2,\ldots,\alpha$ and thus
$r_i$ is not a small packet. Thus for $i=3,\ldots,\alpha$, the amount
of small packets in $(\tau_{i-1},\tau_i]$ is at most
$s(\tau_i-\tau_{i-1})-\ell(r_{i-1})<\ell(r_i)-\ell(r_{i-1})$. The amount
of small packets completed in $(\tau_1,\tau_2]$ is at most
$s(\tau_2-\tau_1)<\ell(r_2)$ and the amount of small packets completed
in $(\tau_\alpha,v]$ is at most $2\ell_k-\ell(r_\alpha)$. Summing this
together, the amount of small packets completed in $(\tau_1,v]$ is at
most $2\ell_k$ and the claim follows. 

\begin{figure}
  \begin{center}
    \includegraphics[width=\textwidth]{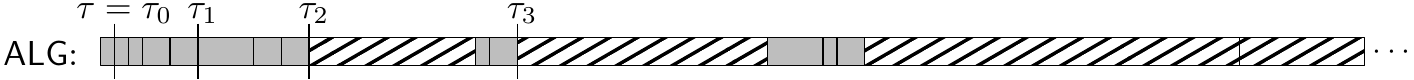}
  \end{center}
  \caption{An illustration of bounding 
  the total size of small packets completed after $\tau$ in the case
  when $\ell_j$ is not pending in the whole block.
  Gray packets are small, while hatched packets have size at least $\ell_j$.
  The times $\tau_1, \tau_2,$ and $\tau_3$ are the ends of phases after $\tau$ (thus $\alpha=3$),
  but $\tau$ need not be the end of a phase.
  }
  \label{fig:s4taus}
\end{figure}

Let $z$ be the total size of packets of size at least
 $\ell_j$ charged by up charges in this block and completed by
$\MAINs$ after $\tau$.
After $\tau$, $\MAINs$ processes packets of total size more than
$s(v-\tau)-\ell_k$ and all of these packets contribute to one of $a$, $o$,
$z$, or the volume of less than $3\ell_k$ of small packets from the
claim above. Thus, using $s\geq4$, we get
\begin{equation}\label{eq:LBon_a}
a>4(v-\tau)-o-z-4\ell_k\,.
\end{equation}

Now we derive two lower bounds on $v-\tau$ using \ADV\ schedule. 
 
Observe that no packet contributing to $z$ except for possibly one (the one
possibly started by $\MAINs$ before $\tau$) is started by \ADV\ before
$\tau$ as otherwise, it would be pending for $\MAINs$ just before $\tau$,
contradicting the definition of $\tau$.

Also, observe that in $(u,\tau]$, \ADV\ runs no packet $p\in P$ with
$\ell(p)>\ell_j$: For a contradiction, assume that such a $p$
exists. As $\tau\leq C_{j'}$ for any $j'\ge j$, such a $p$ is
charged. As $p\in P$, it is charged by a forward charge. However,
then Lemma~\ref{l:forward} implies that at all times between the
start of $p$ in \ADV and $v$ a packet of size $\ell(p)$ is pending
for $\MAINs$; in particular, such a packet is pending in the interval
before $\tau$, contradicting the definition of $\tau$.

These two observations imply that in
$[\tau,v]$, \ADV\ starts and completes all the assigned packets from
$P$, the $n$ packets of size $\ell_j$ from $P$, and all packets except
possibly one contributing to $z$. This gives $v-\tau\geq
\ell(f^{-1}(Q))+n\ell_j+z-\ell_k \geq o/2 +n\ell_j+z-\ell_k$.

To obtain the second bound, we observe that the $n$ packets of size
$\ell_j$ from $P$ are scheduled in $[t,v]$ and together with $t\geq
\tau+3\ell_k$ we obtain $v-\tau=v-t+t-\tau\geq n\ell_j+3\ell_k$.

Summing the two bounds on $v-\tau$ and multiplying by two we get
$4(v-\tau)\geq 4n\ell_j+4\ell_k+o+2z$. Summing with (\ref{eq:LBon_a}) we get
$a>4n\ell_j+z\geq 4n\ell_j$. This completes the proof of the second case.
\end{proof}

As a remark, note that in the previous proof, the first case deals
with blocks after $C_j$, it is the typical and tight case. The second
case deals mainly with the block containing $C_j$, and also with some
blocks before $C_j$, which brings some
technical difficulties, but there is a lot of slack.  This is
similar to the situation in the local analysis using the Master Theorem.

\paragraph{Modifying the adversary schedule}
Now all the packets from $P$ are provisionally assigned by $f$ and for
each $q\in Q$ we have that $\ell(f^{-1}(q)) \le \ell(q)$.

We process each packet $q$ completed by $\MAINs$ in $(u,v]$
according to one of the following three cases; in each case we
remove from \ADV\ one or more packets with total size at most $\ell(q)$.

If $q\not\in Q$, then the definition of $P$ and $Q$ implies that $q$ is
charged by an up charge from some packet $p\not\in P$ of the same
size. We remove $p$ from \ADV.

If $q\in Q$ does not receive a charge, we remove $f^{-1}(q)$ from
\ADV. We have $\ell(f^{-1}(q)) \le \ell(q)$, so the size is as
required. If any packet $p\in f^{-1}(q)$ is charged (necessarily by a
forward charge), we remove this charge.

\begin{figure}
  \begin{center}
    \includegraphics[width=\textwidth]{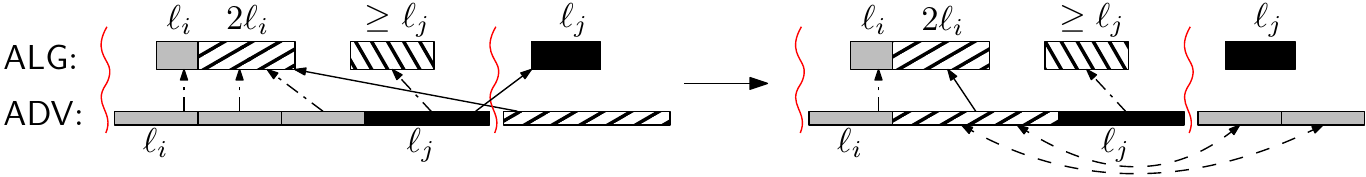}
  \end{center}
  \caption{An illustration of the provisional assignment on the left; note
  that a packet of size $\ell_j$ with a forward charge is also assigned.
  Full arcs depict 1-to-1 charges and dashed arcs depict the provisional assignment.
  The result of modifying the adversary schedule on the right.
  }
  \label{fig:s4assignment}
\end{figure}

If $q\in Q$ receives a charge, it is a back charge from some packet
$p$ of the same size. We remove $p$ from \ADV\ and in the interval
where $p$ was scheduled, we schedule packets from $f^{-1}(q)$ in an
arbitrary order. As $\ell(f^{-1}(q)) \le \ell(q)$, this is
feasible. If any packet $p\in f^{-1}(q)$ is charged, we keep its
charge to the same packet in $\MAINs$; the charge was necessarily a
forward charge, so it leads to some later block.
See Figure~\ref{fig:s4assignment} for an illustration.

After we have processed all the packets $q$, we have modified \ADV\ by
removing an allowed total size of packets and rescheduling the
remaining packets in $(u,v]$ so that any remaining charges go to later
blocks. This completes processing of the block $(u,v]$ and thus also
the proof of 1-competitiveness.

\section{Lower Bounds}\label{sec:LBs}

\subsection{Lower Bound with Two Packet Sizes}

In this section we study lower bounds on the speed necessary to achieve 1-competitiveness.
We start with a lower bound of 2 which holds even for the divisible case.
It follows that our algorithm \DIV\ and the algorithm in Jurdzinski~\etal~\cite{JurdzinskiKL13}
are optimal.
Note that this lower bound follows from results of Anta \etal~\cite{Anta13-dual}
by a similar construction, although the packets in their construction are not released together.

\begin{theorem}\label{thm:LB2}
There is no 1-competitive deterministic online algorithm running with speed $s < 2$,
even if packets have sizes only $1$ and $\ell$ for $\ell > 2s / (2-s)$
and all of them are released at time 0.
\end{theorem}

\begin{proof}
For a contradiction,
consider an algorithm \ALG\ running with speed $s < 2$ that is claimed to be 1-competitive
with an additive constant $A$ where $A$ may depend on $\ell$. 
At time 0 the adversary releases $N_1 = \lceil A / \ell \rceil + 1$ packets of size $\ell$
and $\displaystyle{N_0 = \left\lceil \frac{2\ell}{s} \cdot\big( N_1\cdot (s - 1)\cdot \ell + A + 1\big)\right\rceil}$
packets of size 1. These are all packets in the instance.

The adversary's strategy works by blocks where a block is a time interval 
between two faults and the first block begins at time 0.
The adversary ensures that in each such block \ALG\ completes no packet of size 
$\ell$ and moreover \ADV\ either completes an $\ell$-sized packet, or completes 
more $1$'s (packets of size $1$) than \ALG.

Let $t$ be the time of the last fault; initially $t = 0$.
Let $\tau \ge t$ be the time when \ALG\ starts the first $\ell$-sized packet after $t$ (or at $t$)
if now fault occurs after $t$;
we set $\tau = \infty$ if it does not happen. Note that we use here that \ALG\ is deterministic.
In a block beginning at time $t$,
the adversary proceeds according to the first case below that applies.
\begin{enumerate}[label=(D\arabic*)]
\item If \ADV\ has less than $2\ell / s$ pending packets of size 1, then  \label{p:oneEnd}
the end of the schedule is at $t$.

\item If \ADV\ has all packets of size $\ell$ completed, \label{p:ellEnd}
then it stops the current process and issues faults at times $t+1, t+2, \dots$
Between every two consecutive faults after $t$ it completes one packet of size 1
and it continues issuing faults until 
it has no pending packet of size 1. Then there is the end of the schedule.
Clearly, \ALG\ may complete only packets of size 1 after $t$ as
$\ell > 2s / (2-s) > s$ for $s<2$.

\item If $\tau \ge t+\ell / s - 2$, then the next fault is at time $t+\ell$. \label{p:ell}
In the current block, the adversary completes a packet $\ell$.
\ALG\ completes at most $s\cdot \ell$ packets of size 1 and then it possibly
starts $\ell$ at $\tau$ (if $\tau < t+\ell$) which is jammed,
since it would be completed at
\begin{equation*}
\tau + \frac{\ell}{s} \ge t + \frac{2\ell}{s} - 2 = t + \ell + \left(\frac{2}{s} - 1\right)\ell - 2 > t + \ell
\end{equation*}
where the last inequality follows from $\left(\frac{2}{s} - 1\right)\ell > 2$
which is equivalent to $\ell > 2s / (2-s)$.
Thus the $\ell$-sized packet would be completed after the fault.
See Figure~\ref{fig:lbtwoLate} for an illustration.

\begin{figure}
\centering
\begin{minipage}{.5\textwidth}
    \includegraphics[width=0.95\textwidth]{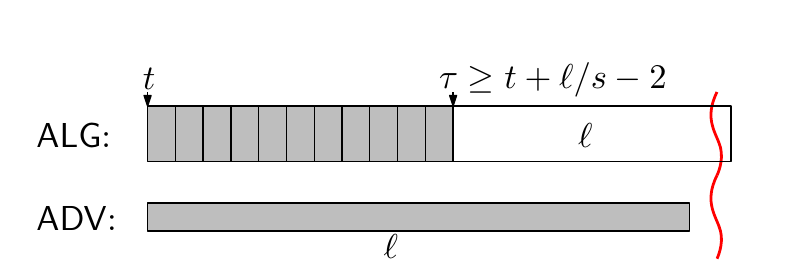}
  \caption{An illustration of Case~\caseref{p:ell}.}
  \label{fig:lbtwoLate}
\end{minipage}%
\begin{minipage}{.5\textwidth}
    \includegraphics[width=0.95\textwidth]{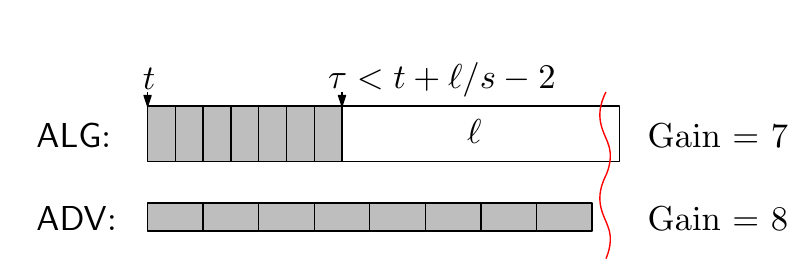}
  \caption{An illustration of Case~\caseref{p:one}.}
  \label{fig:lbtwoEarly}
\end{minipage}
\end{figure}

\item Otherwise, if $\tau < t+\ell / s - 2$, then the next fault
is at time $\tau + \ell / s - \varepsilon$ \label{p:one}
for a small enough $\varepsilon > 0$.
In the current block, \ADV\ completes as many packets of size 1 as it can,
that is $\lfloor \tau + \ell / s - \varepsilon - t\rfloor$ packets of size 1;
note that by Case~\caseref{p:oneEnd}, \ADV\ has enough 1's pending.
Again, the algorithm does not complete the packet of size $\ell$ started at $\tau$,
because it would be finished at $\tau + \ell / s$.
See Figure~\ref{fig:lbtwoEarly} for an illustration.
\end{enumerate}

First notice that the process above ends, since in each block 
the adversary completes a packet. 
We now show $L_{\ADV} > L_{\ALG} + A$ which contradicts the claimed 1-competitiveness of \ALG.

If the adversary's strategy ends in Case~\caseref{p:ellEnd},
then \ADV\ has all $\ell$'s completed and then it schedules all 1's, thus 
$L_{\ADV} = N_1\cdot \ell + N_0 > A + N_0$.
However, \ALG\ does not complete any $\ell$-sized packet and hence $L_{\ALG} \le N_0$
which concludes this case.

Otherwise, the adversary's strategy ends in Case~\caseref{p:oneEnd}.
We first claim that in a block $(t, t']$ created in Case~\caseref{p:one}, \ADV\ 
finishes more 1's than \ALG. Indeed, let $o$ be the number
of 1's completed by \ALG\ in $(t, t']$.
Then $\tau \ge t + o / s$ where $\tau$ is from the adversary's strategy in $(t, t']$,
and we also have $o < \ell - 2s$ or equivalently $\ell > o + 2s$,
because $\tau < t+\ell / s - 2$ in Case~\caseref{p:one}.
The number of 1's scheduled by \ADV\ is
\begin{align}
\left\lfloor \tau + \frac{\ell}{s} - \varepsilon - t\right\rfloor
&\ge \left\lfloor t + \frac{o}{s} + \frac{\ell}{s} - \varepsilon - t\right\rfloor 
\ge \left\lfloor \frac{o}{s} + \frac{o + 2s}{s} - \varepsilon\right\rfloor 
= \left\lfloor \frac{2}{s}o + 2 - \varepsilon\right\rfloor \nonumber \\
&\ge \left\lfloor \frac{2}{s}o + 1\right\rfloor \ge o + 1\nonumber 
\end{align}
and we proved the claim.

Let $\alpha$ be the number of blocks created in Case~\caseref{p:ell};
note that $\alpha\le N_1$, since in each such block \ADV\ finishes one $\ell$-sized packet.
\ALG\ completes at most $s\ell$ packets of size 1 in such a block,
thus $L_{\ADV}((u, v]) - L_{\ALG}((u, v]) \ge (1 - s)\cdot \ell$ for a block $(u, v]$
created in Case~\caseref{p:ell}.

Let $\beta$ be the number of blocks created in Case~\caseref{p:one}.
We have
\begin{equation*}
\beta > \frac{s}{2\ell} \cdot \left(N_0 - \frac{2\ell}{s}\right) = \frac{s\cdot N_0}{2\ell} - 1 = N_1\cdot (s - 1)\cdot \ell + A\,,
\end{equation*}
because in each such block \ADV\ schedules less than
$2\ell / s$ packets of size 1 and less than $2\ell / s$ of these packets are pending at the end.
By the claim above, we have $L_{\ADV}((u, v]) - L_{\ALG}((u, v]) \ge 1$
for a block $(u, v]$ created in Case~\caseref{p:one}.

Summing over all blocks and using the value of $N_0$ we get
\begin{align}
L_{\ADV} - L_{\ALG} &\ge \alpha \cdot (1 - s)\cdot \ell + \beta
>  N_1 \cdot (1 - s)\cdot \ell + N_1\cdot (s - 1)\cdot \ell + A = A\nonumber
\end{align}
where we used $s\ge 1$ which we may suppose w.l.o.g.
This concludes the proof.
\end{proof}

\subsection{Lower Bound for General Packet Sizes}

Our main lower bound of $\phi + 1 = \phi^2 \approx 2.618$
(where $\phi = (\sqrt{5}+1)/2$ is the golden ratio) generalizes the 
construction of Theorem~\ref{thm:LB2}
for more packet sizes, which are no longer divisible.
Still, we make no use of release times.

\begin{theorem}\label{thm:LBphi+1}
There is no 1-competitive deterministic online algorithm running with speed $s < \phi + 1$,
even if all packets are released at time 0.
\end{theorem}

\paragraph{Outline of the proof} We start by describing the adversary's strategy which works against an algorithm
running at speed $s < \phi + 1$, i.e., it shows that it is not $1$-competitive.
It can be seen as a generalization of the strategy with two packet sizes above,
but at the end the adversary sometimes needs a new strategy how to complete
all short packets (of size less than $\ell_i$ for some $i$), preventing the algorithm to complete a
long packet (of size at least $\ell_i$).

Then we show a few lemmata about the behavior of the algorithm.
Finally, we prove that the gain of the adversary, i.e., the total size of its completed packets,
is substantially larger than the gain of the algorithm.

\paragraph{Adversary's strategy}
The adversary chooses $\varepsilon > 0$ small enough and $k\in \N$ large enough
so that $s < \phi + 1 - 1 / \phi^{k-1}$.
For convenience, the smallest size in the instance is $\varepsilon$ instead of $1$.
There will be $k+1$ packet sizes in the instance, namely $\ell_0 = \varepsilon$,
and $\ell_i = \phi^{i-1}$ for $i = 1, \dots, k$.

Suppose for a contradiction that there is an algorithm \ALG\ running with speed $s < \phi + 1$
that is claimed to be 1-competitive with an additive constant $A$ where $A$ may depend on $\ell_i$'s,
in particular on $\varepsilon$ and $k$.
The adversary issues $N_i$ packets of size $\ell_i$ at time 0, for $i = 0, \dots, k$;
$N_i$'s are chosen so that $N_0\gg N_1\gg \dots \gg N_k$. These are all packets in the instance.

More precisely, $N_i$'s are defined inductively so that it holds that $N_k > A / \ell_k$, 
$N_i > \phi s \ell_k  \sum_{j = i+1}^k N_j + A / \ell_i$ for $i = k-1, \dots, 1$,
and finally $$N_0 > \frac{A + 1 + \phi \ell_k}{\varepsilon^2} \cdot \big(\phi s \ell_k \sum_{i=1}^k N_i\big)\,.$$

The adversary's strategy works 
by blocks where a block is again a time interval between two faults
and the first block begins at time 0.
Let $t$ be the time of the last fault; initially $t = 0$.
Let $\tau_i \ge t$ be the time when \ALG\ starts the first packet of size $\ell_i$ after $t$ (or at $t$)
if no fault occurs after $t$;
we set $\tau_i = \infty$ if it does not happen. Again, we use here that \ALG\ is deterministic.
Let $\tau_{\ge i} = \min_{j\ge i}\tau_j$ be the time when \ALG\ starts the first packet 
of size at least $\ell_i$ after $t$.
Let $P_{\ADV}(i)$ be the total size of $\ell_i$'s (packets of size $\ell_i$)
pending for the adversary at time $t$.

In a block beginning at time $t$,
the adversary proceeds according to the first case below that applies.
Each case has an intuitive explanation which we make precise later.
\begin{enumerate}[label=(B\arabic*)]
\item  \label{case:epsEnd} If there are less than $\phi\ell_k / \varepsilon$
packets of size $\varepsilon$ pending for \ADV, then
the end of the schedule is at time $t$.

Lemma~\ref{l:moreEpsilons} below shows that in blocks in which \ADV\ schedules $\varepsilon$'s
it completes more than \ALG\ in terms of total size.
It follows that the schedule of \ADV\ has much larger total completed size 
for $N_0$ large enough, since the adversary scheduled nearly all packets of size $\varepsilon$;
see Lemma~\ref{l:epsEnd}.

\item \label{case:iEnd} If there is $i \ge 1$ such that $P_{\ADV}(i) = 0$, then \ADV\ stops the current process
and continues by Strategy \textsc{Finish} described below.

\item \label{case:tau1tooEarly} If $\tau_1 < t + \ell_1 / (\phi\cdot s)$, then the next fault occurs
at time $\tau_1 + \ell_1 / s - \varepsilon$,
so that \ALG\ does not finish the first $\ell_1$-sized packet. \ADV\ schedules as many
$\varepsilon$'s as it can.

In this case, \ALG\ schedules $\ell_1$ too early and
in Lemma~\ref{l:moreEpsilons} we show that the total size of packets completed by \ADV\ is larger
than the total size of packets completed by \ALG.

\item \label{case:tau2tooEarly} If $\tau_{\ge 2} < t + \ell_2 / (\phi\cdot s)$, then the next fault is at time
$\tau_{\ge 2} + \ell_2 / s - \varepsilon$,
so that \ALG\ does not finish the first packet of size at least $\ell_2$.
\ADV\ again schedules as many $\varepsilon$'s as it can.
Similarly as in the previous case, \ALG\ starts $\ell_2$ or a larger packet too early
and we show that \ADV\ completes more in terms of size than \ALG, again using Lemma~\ref{l:moreEpsilons}.

\item \label{case:tauI+1beforeTauI} If there is $1\le i < k$ such that $\tau_{\ge i+1} < \tau_i$,
then we choose the smallest such $i$ and
the next fault is at time $t + \ell_i$. 
\ADV schedules a packet of size $\ell_i$.
See Figure~\ref{fig:lbphitauI+1beforeTauI} for an illustration.

Intuitively, this case means that \ALG\ skips $\ell_i$ and
schedules $\ell_{i+1}$ (or a larger packet) earlier.
Lemma~\ref{l:tauI+1beforeTauI} shows that the algorithm cannot finish its first packet of size
at least $\ell_{i+1}$ (thus it also does not schedule $\ell_i$)
provided that this case is not triggered for a smaller $i$, or previous cases are not triggered.

\begin{figure}
\centering
\begin{minipage}{.38\textwidth}
    \includegraphics[width=\textwidth]{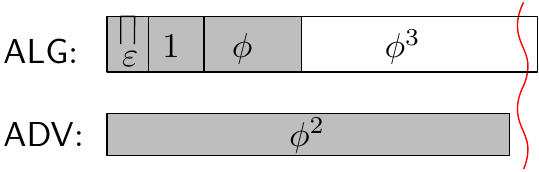}
  \caption{An illustration of Case~\caseref{case:tauI+1beforeTauI}.}
  \label{fig:lbphitauI+1beforeTauI}
\end{minipage}%
\begin{minipage}{.06\textwidth}
~~~~
\end{minipage}%
\begin{minipage}{.55\textwidth}
    \includegraphics[width=\textwidth]{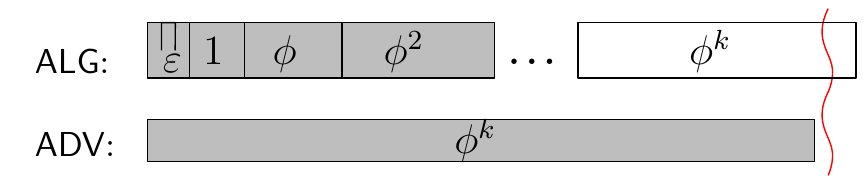}
  \caption{An illustration of Case~\caseref{case:tauKmustBeLate}.}
  \label{fig:lbphiTauKlate}
\end{minipage}
\end{figure}

\item \label{case:tauKmustBeLate} Otherwise, the next fault occurs at $t+\ell_k$
and \ADV\ schedules a packet of size $\ell_k$ in this block. Lemma~\ref{l:tauKmustBeLate} shows that \ALG\ cannot
complete an $\ell_k$-sized packet in this block.
See Figure~\ref{fig:lbphiTauKlate} for an illustration.
\end{enumerate}

We remark that the process above eventually ends
either in Case~\caseref{case:epsEnd}, or in Case~\caseref{case:iEnd},
since in each block \ADV\ schedules a packet.
Also note that the length of each block is at most $\phi\ell_k$.

We describe Strategy \textsc{Finish}, started in Case~\caseref{case:iEnd}. 
Let $i$ be the smallest index $i'\ge 1$ such that $P_{\ADV}(i') = 0$.
For brevity, we call a packet of size at least $\ell_i$ \textit{long},
and a packet of size $\ell_j$ with $1\le j < i$ \textit{short};
note that $\varepsilon$'s are not short packets.
In a nutshell, \ADV\ tries to schedule all short packets,
while preventing the algorithm from completing any long packet.
Similarly to Cases~\caseref{case:tau1tooEarly} and~\caseref{case:tau2tooEarly},
if \ALG\ is starting a long packet too early, 
\ADV\ schedules $\varepsilon$'s and gains in terms of total size.

Adversary's Strategy \textsc{Finish} works again by blocks.
Let $t$ be the time of the last fault.
Let $\tau \ge t$ be the time when \ALG\ starts
the first long packet after $t$;
we set $\tau = \infty$ if it does not happen.
The adversary proceeds according to the first case below that applies:
\begin{enumerate}[label=(F\arabic*)]
\item  \label{case2:epsEnd}
If $P_{\ADV}(0) < \phi\ell_k$, then
the end of the schedule is at time $t$.

\item \label{case2:bigEnd}
If \ADV\ has no pending short packet, then the strategy \textsc{Finish} ends
and \ADV\ issues faults at times $t+\varepsilon, t+2\varepsilon, \dots$
Between every two consecutive faults after $t$ it completes one packet of size $\varepsilon$
and it continues issuing faults until 
it has no pending $\varepsilon$. Then there is the end of the schedule.
Clearly, \ALG\ may complete only $\varepsilon$'s after $t$ if $\varepsilon$ is small enough.
Note that for $i=1$ this case is immediately triggered, as $\ell_0$-sized packets
are not short, hence there are no short packets whatsoever.

\item \label{case2:tauTooEarly}
If $\tau < t + \ell_i / (\phi\cdot s)$, then the next fault is at time
$\tau + \ell_i / s - \varepsilon$,
so that \ALG\ does not finish the first long packet.
\ADV\ schedules as many $\varepsilon$'s as it can.
Note that the length of this block is less than $\ell_i / (\phi\cdot s) + \ell_i / s \le \phi \ell_k$.
Again, we show that \ADV\ completes more in terms of size using Lemma~\ref{l:moreEpsilons}.

\item \label{case2:regular}
Otherwise, $\tau \ge t + \ell_i / (\phi\cdot s)$.
\ADV\ issues the next fault at time $t+\ell_{i-1}$.
Let $j$ be the largest $j' < i$ such that $P_{\ADV}(j') > 0$.
\ADV\ schedules a packet of size $\ell_j$ which is completed as $j\le i-1$.
Lemma~\ref{l:FinishRegularCase} shows that \ALG\ does not complete the long packet started at $\tau$. 
\end{enumerate}

Again, in each block \ADV\ completes a packet, thus 
Strategy \textsc{Finish} eventually ends.
Note that the length of each block is less than $\phi \ell_k$.

\paragraph{Properties of the adversary's strategy}
We now prove the lemmata mentioned above. In the following, $t$ is
the beginning of the considered block and $t'$ is the end of the block, i.e., the time of the next fault after $t$.
Recall that $L_{\ALG}((t, t'])$ is the total size of packets completed by \ALG\ in $(t, t']$.
We start with a general lemma that covers all cases in which \ADV\ schedules many $\varepsilon$'s.

\begin{lemma}\label{l:moreEpsilons}
In Cases~\caseref{case:tau1tooEarly}, \caseref{case:tau2tooEarly}, and~\caseref{case2:tauTooEarly},
$L_{\ADV}((t, t']) \ge L_{\ALG}((t, t']) + \varepsilon$ holds.
\end{lemma}

\begin{proof}
Let $i$ and $\tau$ be as in Case~\caseref{case2:tauTooEarly};
we set $i=1$ and $\tau = \tau_1$ in Case~\caseref{case:tau1tooEarly},
and $i=2$ and $\tau=\tau_{\ge 2}$ in Case~\caseref{case:tau2tooEarly}.
Note that the first packet of size (at least) $\ell_i$ is started at $\tau$
with $\tau < t + \ell_i / (\phi\cdot s)$
and that the next fault occurs at time $\tau + \ell_i / s - \varepsilon$.
Furthermore, $P_{\ADV}(0, t) \ge \phi\ell_k$ by Cases~\caseref{case:epsEnd} and~\caseref{case2:epsEnd}.
As $t' - t\le \phi\ell_k$ it follows that \ADV\ has enough
$\varepsilon$'s to fill nearly the whole block with them, so
in particular $L_{\ADV}((t, t']) > t' - t - \varepsilon$.

Let $a = L_{\ALG}((t, t'])$. Since \ALG\ does not complete the 
$\ell_i$-sized packet we have $\tau \ge t + a / s$ and thus also $a < \ell_i / \phi$
as $\tau < t + \ell_i / (\phi\cdot s)$.

If $a < \ell_i / \phi - 3s\varepsilon / \phi$
which is equivalent to $\ell_i > \phi\cdot a + 3s\varepsilon$, then
we show the required inequality by the following calculation:
\begin{equation*}
L_{\ADV}((t, t']) + \varepsilon > t' - t
= \tau + \frac{\ell_i}{s} - \varepsilon - t
\ge \frac{a}{s} + \frac{\ell_i}{s} - \varepsilon
> \frac{a + \phi\cdot a + 3s\varepsilon}{s} - \varepsilon
> a+2\varepsilon\,,
\end{equation*}
where the last inequality follows from $s<\phi+1$.

Otherwise, $a$ is nearly $\ell_i / \phi$ and thus large enough. Then we get
\begin{equation*}
L_{\ADV}((t, t']) + \varepsilon  > t' - t = \tau + \frac{\ell_i}{s} - \varepsilon - t 
\ge \frac{a}{s} + \frac{\ell_i}{s} - \varepsilon 
> \frac{a}{s} + \frac{\phi a}{s} - \varepsilon 
> a+2\varepsilon 
\end{equation*}
where the penultimate inequality follows by $\ell_i > \phi a$, 
and the last inequality holds as $(1+\phi) a / s > a + 3\varepsilon$
for $\varepsilon$ small enough and $a \geq \ell_i/\phi - 3s\varepsilon / \phi$.
\end{proof}

For brevity, we inductively define $S_0 = \phi - 1$ and $S_i = S_{i-1} + \ell_i$ for $i=1, \dots, k$.
Thus $S_i = \sum_{j=1}^{i} \ell_i + \phi - 1$ and a calculation shows $S_i = \phi^{i+1} - 1$.
We prove a useful observation.

\begin{lemma}\label{l:LBonTauI}
Fix $j\ge 2$.
If Case~\caseref{case:tau1tooEarly} and
Case~\caseref{case:tauI+1beforeTauI} for $i < j$ are not triggered in the block,
then $\tau_{i'+1}\ge t + S_{i'} / s$ for each $i' < j$.
\end{lemma}

\begin{proof}
We have $\tau_1\ge t + \ell_1 / (\phi\cdot s) = t + (\phi - 1) / s$ by Case~\caseref{case:tau1tooEarly}
and $\tau_{i+1} \ge \tau_i + \ell_i / s$ for any $i < j$, since Case~\caseref{case:tauI+1beforeTauI}
was not triggered for $i < j$ and the first $\ell_i$-sized packet needs to be finished before starting the next packet.
Summing the bounds gives the inequalities in the lemma.
\end{proof}


\begin{lemma}\label{l:tauI+1beforeTauI}
In Case~\caseref{case:tauI+1beforeTauI},
the algorithm does not complete any packet of size $\ell_i$ or larger.
\end{lemma}

\begin{proof}
Recall that we have $\tau_{\ge i+1} < \tau_i$,
thus the first started packet $p$ of size at least $\ell_i$ has size at least $\ell_{i+1}$.
It suffices to prove
\begin{equation}\label{eq:tauI+1LB}
\tau_{\ge i+1} + \frac{\ell_{i+1}}{s} - t > \ell_i\,,
\end{equation}
which means that $p$ would be completed after the next fault at time $t + \ell_i$.

We start with the case $i=1$ in which $\tau_{\ge 2} < \tau_1$. Since Case~\caseref{case:tau2tooEarly} was not
triggered, we have $\tau_{\ge 2} \ge t + \ell_2 / (\phi\cdot s) = t + 1 / s$. We show (\ref{eq:tauI+1LB}) by the following calculation:
\begin{equation*}
\tau_{\ge 2} + \frac{\ell_2}{s} - t \ge \frac{1}{s} + \frac{\ell_2}{s}
= \frac{1 + \phi}{s} > 1 = \ell_1 \enspace,
\end{equation*}
where the strict inequality holds by $s < \phi + 1$.

Now consider the case $i \ge 2$. By the minimality of $i$ satisfying
the condition of Case~\caseref{case:tauI+1beforeTauI},
we use Lemma~\ref{l:LBonTauI} for $j=i$ and $i'=i-2$ to get
$\tau_{i-1}\ge t + S_{i-2} / s$.
Since a packet $\ell_{i-1}$ is started at $\tau_{i-1}$ and must be finished by $\tau_{\ge i+1}$, it holds
$\tau_{\ge i+1}\ge t+(S_{i-2} + \ell_{i-1}) / s = t+S_{i-1} / s$. Thus
\begin{align*}
\tau_{\ge i+1} + \frac{\ell_{i+1}}{s} - t 
\ge \frac{S_{i-1} + \ell_{i+1}}{s} 
&= \frac{\phi^i - 1 + \phi^i}{s} \\
&= \frac{\phi^{i + 1} + \phi^{i-2} - 1}{s} 
\ge \frac{\phi^{i + 1}}{s} 
> \phi^{i-1} = \ell_i\,, 
\end{align*}
where the penultimate inequality holds by $i\ge 2$ and 
the last inequality by $s < \phi + 1$. 
(We remark that the penultimate inequality has a significant slack for $i > 2$.)
\end{proof}

\begin{lemma}\label{l:tauKmustBeLate}
In Case~\caseref{case:tauKmustBeLate}, \ALG\ does not complete a packet of size $\ell_k$.
\end{lemma}

\begin{proof}
It suffices to prove
\begin{equation}\label{eq:tauKlate}
\tau_k > t + \left(1-\frac{1}{s}\right)\ell_k\,,
\end{equation}
since then \ALG\ completes the first $\ell_k$-sized packet at
\begin{equation*}
\tau_k + \frac{\ell_k}{s} > t + \left(1-\frac{1}{s}\right)\ell_k + \frac{\ell_k}{s} = t + \ell_k\,,
\end{equation*}
i.e., after the next fault at time $t+\ell_k$.

Recall that we choose $k$ large enough so that $s < \phi + 1 - 1 / \phi^{k-1}$
or equivalently $\phi - 1/\phi^{k-1} > s - 1$.
We multiply the inequality by $\phi^{k-1}$, divide it by $s$ and add $t$ to both sides and we get
\begin{equation}\label{eq:tauKlate-2}
t + \frac{\phi^k - 1}{s} > t + \left(1-\frac{1}{s}\right) \phi^{k-1} = t + \left(1-\frac{1}{s}\right) \ell_k\,.
\end{equation}
Since Cases~\caseref{case:tau1tooEarly} and~\caseref{case:tauI+1beforeTauI}
are not triggered, we use Lemma~\ref{l:LBonTauI} for $j = k$ to show 
$\tau_k\ge t + S_{k-1} / s = t + (\phi^k - 1) / s$. We combine this with \eqref{eq:tauKlate-2}
and we have 
\begin{equation}
\tau_k\ge t + \frac{\phi^k - 1}{s} > t + \left(1-\frac{1}{s}\right) \ell_k\,,
\end{equation}
which shows \eqref{eq:tauKlate} and concludes the proof of the lemma.
%
\end{proof}

\begin{lemma}\label{l:FinishRegularCase}
In Case~\caseref{case2:regular}, \ALG does not complete any long packet.
\end{lemma}

\begin{proof}
Recall that the first long packet $p$ is started at $\tau$
and it has size of at least $\ell_i$, thus 
it would be completed at $\tau + \ell_i / s$ or later.
We show $\tau + \ell_i / s - t > \ell_{i-1}$ by the following calculation:
\begin{equation*}
\tau + \frac{\ell_i}{s} - t \ge \frac{\ell_i}{\phi\cdot s} + \frac{\ell_i}{s}
= \frac{\phi \ell_i}{s} > \frac{\ell_i}{\phi} = \ell_{i-1}\,,
\end{equation*}
where the strict inequality holds by $s < \phi + 1$.
This implies that the long packet $p$ would be completed after
the next fault at time $t + \ell_{i-1}$.
\end{proof}

\paragraph{Analysis of the gains}
We are ready to prove that at the end of the schedule
$L_{\ADV} > L_{\ALG} + A$ holds,
which contradicts the claimed 1-competitiveness of \ALG
and proves Theorem~\ref{thm:LBphi+1}.
We inspect all the cases in which the instances may end,
starting with Cases~\caseref{case:epsEnd} and~\caseref{case2:epsEnd}.
We remark that we use only crude bounds to keep the analysis simple.

\begin{lemma}\label{l:epsEnd}
If the schedule ends in Case~\caseref{case:epsEnd} or~\caseref{case2:epsEnd},
we have $L_{\ADV} > L_{\ALG} + A$.
\end{lemma}

\begin{proof}
Recall that each block $(t, t']$ has length of at most $\phi \ell_k$,
thus $L_{\ALG}((t, t'])\le s\phi \ell_k$ and
$L_{\ADV}((t, t'])\le \phi \ell_k$.

We call a block in which \ADV\ schedules many $\varepsilon$'s \textit{small},
other blocks are \textit{big}. Recall that \ADV schedules no $\varepsilon$ in a big block. Note that 
Cases~\caseref{case:tau1tooEarly}, \caseref{case:tau2tooEarly}, and~\caseref{case2:tauTooEarly}
concern small blocks, whereas
Cases~\caseref{case:tauI+1beforeTauI}, \caseref{case:tauKmustBeLate}, and~\caseref{case2:regular}
concern big blocks.

By Lemma~\ref{l:moreEpsilons}, in each small block $(t, t']$
it holds that $L_{\ADV}((t, t']) \ge L_{\ALG}((t, t']) + \varepsilon$.
Let $\beta$ be the number of small blocks.
We observe that
\begin{equation*}
\beta \ge \frac{\left(N_0 - \frac{\phi\ell_k}{\varepsilon}\right) \varepsilon}{\phi \ell_k}\,,
\end{equation*}
because in each such block \ADV\ schedules at most $\phi \ell_k / \varepsilon$ packets of size $\varepsilon$
and $P_{\ADV}(0) < \phi\ell_k$ at the end in Cases~\caseref{case:epsEnd} and~\caseref{case2:epsEnd}.

The number of big blocks is at most $\sum_{i=1}^k N_i$,
since in each such block \ADV\ schedules a packet of size at least $\ell_1$.
For each such block we have $L_{\ADV}((t, t']) - L_{\ALG}((t, t'])\ge -s\phi \ell_k$
which is only a crude bound, but it suffices for $N_0$ large enough.

Summing over all blocks we obtain
\begin{align}
L_{\ADV} - L_{\ALG} &\ge \beta \varepsilon - \phi s \ell_k \sum_{i=1}^k N_i 
\ge \frac{\left(N_0 - \frac{\phi\ell_k}{\varepsilon}\right) \varepsilon^2}{\phi \ell_k}  - \phi s \ell_k \sum_{i=1}^k N_i \nonumber \\
&> A + \phi s \ell_k \sum_{i=1}^k N_i  - \phi s \ell_k \sum_{i=1}^k N_i  = A \enspace, \label{eqn:epsEndLast}
\end{align}
where (\ref{eqn:epsEndLast}) follows from $N_0 > \phi \ell_k (A + 1 + \phi s \ell_k \sum_{i=1}^k N_i) / \varepsilon^2$.
\end{proof}

It remains to prove the same for termination by Case~\caseref{case2:bigEnd},
since there is no other case in which the strategy may end.

\begin{lemma}\label{l:Finish}
If Strategy \textsc{Finish} ends in Case~\caseref{case2:bigEnd},
then $L_{\ADV} > L_{\ALG} + A$.
\end{lemma}

\begin{proof}
Note that \ADV\ schedules all short packets and all $\varepsilon$'s, i.e., those of size less than 
$\ell_i$.  In particular, we have $L_{\ADV}(<i) \ge L_{\ALG}(<i)$.

Call a block in which \ALG\ completes a packet of size at least $\ell_i$ \textit{bad}.
As the length of any block is at most $\phi\ell_k$ we get that $L_{\ALG}(\ge i, (t, t']) \le s\phi \ell_k$
for a bad block $(t, t']$.
Bad blocks are created only in Cases~\caseref{case:tauI+1beforeTauI} and~\caseref{case:tauKmustBeLate},
but in each bad block \ADV\ finishes a packet strictly larger than $\ell_i$;
note that here we use Lemmata~\ref{l:tauI+1beforeTauI} and~\ref{l:tauKmustBeLate}.
Hence the number of bad blocks is bounded by $\sum_{j=i+1}^{k} N_j$. 
As \ADV\ completes all $\ell_i$'s we obtain
\begin{align*}
L_{\ADV}(\ge i) - L_{\ALG}(\ge i)
&\ge \ell_i N_i - \phi s \ell_k \sum_{j=i+1}^{k} N_j\\
&> \ell_i \phi s \ell_k \sum_{j = i+1}^k N_j + A - \phi s \ell_k \sum_{j=i+1}^{k} N_j \ge A\,,
\end{align*}
where the strict inequality follows from $N_k > A / \ell_k$ for $i = k$
and from $N_i > \phi s \ell_k \sum_{j = i+1}^k N_j + A / \ell_i$ for $i < k$.
By summing it with $L_{\ADV}(<i) \ge L_{\ALG}(<i)$ we conclude that $L_{\ADV} > L_{\ALG} + A$.
\end{proof}



\begin{thebibliography}{10}

\bibitem{Anta13}
Antonio~Fern{\'{a}}ndez Anta, Chryssis Georgiou, Dariusz~R. Kowalski, Joerg
  Widmer, and Elli Zavou.
\newblock Measuring the impact of adversarial errors on packet scheduling
  strategies.
\newblock {\em J. Scheduling}, 19(2):135--152, 2016.
\newblock Also appeared in Proc.\ of SIROCCO 2013: 261--273.
\newblock \href {http://dx.doi.org/10.1007/s10951-015-0451-z}
  {\path{doi:10.1007/s10951-015-0451-z}}.

\bibitem{Anta13-dual}
Antonio~Fern{\'{a}}ndez Anta, Chryssis Georgiou, Dariusz~R. Kowalski, and Elli
  Zavou.
\newblock Online parallel scheduling of non-uniform tasks: Trading failures for
  energy.
\newblock {\em Theor. Comput. Sci.}, 590:129--146, 2015.
\newblock Also appeared in Proc.\ of FCT 2013: 145-158.
\newblock \href {http://dx.doi.org/10.1016/j.tcs.2015.01.027}
  {\path{doi:10.1016/j.tcs.2015.01.027}}.

\bibitem{AntaGKZ15}
Antonio~Fern{\'{a}}ndez Anta, Chryssis Georgiou, Dariusz~R. Kowalski, and Elli
  Zavou.
\newblock Competitive analysis of fundamental scheduling algorithms on a
  fault-prone machine and the impact of resource augmentation.
\newblock {\em Future Generation Comp. Syst.}, 78:245--256, 2018.
\newblock Also appeared in Proc.\ of ARMS-CC@PODC 2015, LNCS 9438: 1--16.
\newblock \href {http://dx.doi.org/10.1016/j.future.2016.05.042}
  {\path{doi:10.1016/j.future.2016.05.042}}.

\bibitem{Ben-DavidBKTW94}
Shai Ben{-}David, Allan Borodin, Richard~M. Karp, G{\'{a}}bor Tardos, and Avi
  Wigderson.
\newblock On the power of randomization in on-line algorithms.
\newblock {\em Algorithmica}, 11(1):2--14, 1994.
\newblock Also appeared in Proc.\ of STOC 1990: 379--386.
\newblock \href {http://dx.doi.org/10.1007/BF01294260}
  {\path{doi:10.1007/BF01294260}}.

\bibitem{jamming-waoa}
Martin B{\"{o}}hm, {\L}ukasz Je{\.z}, Ji{\v{r}}{\'{\i}} Sgall, and Pavel
  Vesel{\'{y}}.
\newblock On packet scheduling with adversarial jamming and speedup.
\newblock In {\em Proc.\ of the 15th Workshop on Approximation and Online
  Algorithms ({WAOA'17})}, volume 10787 of {\em LNCS}, pages 190--206, 2018.
\newblock \href {http://dx.doi.org/10.1007/978-3-319-89441-6_15}
  {\path{doi:10.1007/978-3-319-89441-6_15}}.

\bibitem{BEY}
Allan Borodin and Ran El{-}Yaniv.
\newblock {\em Online computation and competitive analysis}.
\newblock Cambridge University Press, 1998.

\bibitem{Chrobak-overloaded}
Marek Chrobak, Leah Epstein, John Noga, Ji{\v{r}}{\'{\i}} Sgall, Rob van Stee,
  Tom{\'{a}}{\v{s}} Tich{\'{y}}, and Nodari Vakhania.
\newblock Preemptive scheduling in overloaded systems.
\newblock {\em J. Comput. Syst. Sci.}, 67:183--197, 2003.
\newblock \href {http://dx.doi.org/10.1016/S0022-0000(03)00070-9}
  {\path{doi:10.1016/S0022-0000(03)00070-9}}.

\bibitem{GarncarekJL17}
P.~Garncarek, T.~Jurdziński, and K.~Loryś.
\newblock Fault-tolerant online packet scheduling on parallel channels.
\newblock In {\em 2017 IEEE International Parallel and Distributed Processing
  Symposium (IPDPS)}, pages 347--356, May 2017.
\newblock \href {http://dx.doi.org/10.1109/IPDPS.2017.105}
  {\path{doi:10.1109/IPDPS.2017.105}}.

\bibitem{GeorgiouK15}
Chryssis Georgiou and Dariusz~R. Kowalski.
\newblock On the competitiveness of scheduling dynamically injected tasks on
  processes prone to crashes and restarts.
\newblock {\em J. Parallel Distrib. Comput.}, 84:94--107, 2015.
\newblock \href {http://dx.doi.org/10.1016/j.jpdc.2015.07.007}
  {\path{doi:10.1016/j.jpdc.2015.07.007}}.

\bibitem{JurdzinskiKL13}
Tomasz Jurdzinski, Dariusz~R. Kowalski, and Krzysztof Loryś.
\newblock Online packet scheduling under adversarial jamming.
\newblock In {\em Proc.\ of the 12th Workshop on Approximation and Online
  Algorithms (WAOA'14)}, volume 8952 of {\em LNCS}, pages 193--206, 2015.
\newblock See \url{http://arxiv.org/abs/1310.4935} for missing proofs.
\newblock \href {http://dx.doi.org/10.1007/978-3-319-18263-6_17}
  {\path{doi:10.1007/978-3-319-18263-6_17}}.

\bibitem{speed-clairvoyance}
Bala Kalyanasundaram and Kirk Pruhs.
\newblock Speed is as powerful as clairvoyance.
\newblock {\em J.\ of the ACM}, 47(4):617--643, 2000.
\newblock Also appeared in Proc.\ of FOCS 1995: 214--221.
\newblock \href {http://dx.doi.org/10.1145/347476.347479}
  {\path{doi:10.1145/347476.347479}}.

\bibitem{Koo-tight-deadlines}
Chiu{-}Yuen Koo, Tak~Wah Lam, Tsuen{-}Wan Ngan, Kunihiko Sadakane, and
  Kar{-}Keung To.
\newblock On-line scheduling with tight deadlines.
\newblock {\em Theor. Comput. Sci.}, 295:251--261, 2003.
\newblock \href {http://dx.doi.org/10.1016/S0304-3975(02)00407-3}
  {\path{doi:10.1016/S0304-3975(02)00407-3}}.

\bibitem{kowalski_fault_tolerant_2_sizes_17}
Dariusz~R. Kowalski, Prudence W.~H. Wong, and Elli Zavou.
\newblock Fault tolerant scheduling of tasks of two sizes under resource
  augmentation.
\newblock {\em Journal of Scheduling}, 20(6):695--711, Dec 2017.
\newblock \href {http://dx.doi.org/10.1007/s10951-017-0541-1}
  {\path{doi:10.1007/s10951-017-0541-1}}.

\bibitem{Lam-EDF-overload}
Tak~Wah Lam, Tsuen{-}Wan Ngan, and Kar{-}Keung To.
\newblock Performance guarantee for {EDF} under overload.
\newblock {\em J. Algorithms}, 52:193--206, 2004.
\newblock \href {http://dx.doi.org/10.1016/j.jalgor.2003.10.004}
  {\path{doi:10.1016/j.jalgor.2003.10.004}}.

\bibitem{Lam-tradeoffs}
Tak~Wah Lam and Kar{-}Keung To.
\newblock Trade-offs between speed and processor in hard-deadline scheduling.
\newblock In {\em Proc.\ of the 10th Annual {ACM-SIAM} Symposium on Discrete
  Algorithms (SODA'99)}, pages 623--632. {ACM/SIAM}, 1999.
\newblock URL: \url{http://dl.acm.org/citation.cfm?id=314500.314884}.

\bibitem{Phillips-time-critical}
Cynthia~A. Phillips, Clifford Stein, Eric Torng, and Joel Wein.
\newblock Optimal time-critical scheduling via resource augmentation.
\newblock {\em Algorithmica}, 32:163--200, 2002.
\newblock Also appeared in Proc.\ of STOC 1997: 140--149.
\newblock \href {http://dx.doi.org/10.1007/s00453-001-0068-9}
  {\path{doi:10.1007/s00453-001-0068-9}}.

\bibitem{Pruhs07}
Kirk Pruhs.
\newblock Competitive online scheduling for server systems.
\newblock {\em {SIGMETRICS} Performance Evaluation Review}, 34(4):52--58, 2007.
\newblock \href {http://dx.doi.org/10.1145/1243401.1243411}
  {\path{doi:10.1145/1243401.1243411}}.

\bibitem{Schewior-thesis}
Kevin Schewior.
\newblock {\em Handling Critical Tasks Online: Deadline Scheduling and
  Convex-Body Chasing}.
\newblock {PhD} dissertation, TU Berlin, 2016.
\newblock URL: \url{http://dx.doi.org/10.14279/depositonce-5427}.

\bibitem{SleatorT85}
Daniel~Dominic Sleator and Robert~Endre Tarjan.
\newblock Amortized efficiency of list update and paging rules.
\newblock {\em Commun. {ACM}}, 28(2):202--208, 1985.
\newblock \href {http://dx.doi.org/10.1145/2786.2793}
  {\path{doi:10.1145/2786.2793}}.

\end{thebibliography}

\end{document}